\def\year{2021}\relax
\documentclass[letterpaper]{article} % DO NOT CHANGE THIS
\usepackage{aaai21}  % DO NOT CHANGE THIS
\usepackage{times}  % DO NOT CHANGE THIS
\usepackage{helvet} % DO NOT CHANGE THIS
\usepackage{courier}  % DO NOT CHANGE THIS
\usepackage[hyphens]{url}  % DO NOT CHANGE THIS
\usepackage{graphicx} % DO NOT CHANGE THIS
\usepackage{natbib}  % DO NOT CHANGE THIS AND DO NOT ADD ANY OPTIONS TO IT
\usepackage{caption} % DO NOT CHANGE THIS AND DO NOT ADD ANY OPTIONS TO IT
\title{Deep Portfolio Optimization
via Distributional Prediction of Residual Factors}

\author {
    Kentaro Imajo, \textsuperscript{\rm 1}
    Kentaro Minami, \textsuperscript{\rm 1}
    Katsuya Ito, \textsuperscript{\rm 1}
    Kei Nakagawa, \textsuperscript{\rm 2}
    \\
}
\affiliations {
    \textsuperscript{\rm 1} Preferred Networks, Inc. \\
    \textsuperscript{\rm 2} Nomura Asset Management Co., Ltd. \\
    \{imos, minami, katsuya1ito\}@preferred.jp,
    k-nakagawa@nomura-am.co.jp
}

\usepackage[switch]{lineno} % line numbers

\usepackage{bm}
\usepackage{booktabs}

\usepackage{tikz}
\usetikzlibrary{positioning}

\usepackage{enumitem}

\usepackage{amsthm}
\theoremstyle{definition}
\newtheorem{definition}{Definition}
\newtheorem{rmk}{Remark}
\newtheorem{prop}{Proposition}
\newtheorem{lem}{Lemma}

\usepackage{latexsym}
\usepackage{amsmath}
\usepackage{amssymb}
\usepackage{color}

\newcommand{\argmin}{\operatornamewithlimits{argmin}}

\newcommand{\EE}{\mathbb{E}}

\newcommand{\RR}{\mathbb{R}}

\newcommand{\mcH}{\mathcal{H}}

\newcommand{\mcL}{\mathcal{L}}

\newcommand{\mcV}{\mathcal{V}}

\newcommand{\norm}[1]{\lVert #1 \rVert}
\newcommand{\Var}{\mathrm{Var}}

\newcommand{\set}[1]{\{#1\}}

\newcommand{\diag}{\mathord{\mathrm{diag}}}

\newcommand{\com}[1]{#1}

\begin{document}
%\linenumbers
\maketitle

\begin{abstract}
Recent developments in deep learning techniques have motivated intensive research in machine learning-aided stock trading strategies.
However, since the financial market has a highly non-stationary nature hindering the application of typical data-hungry machine learning methods, leveraging financial inductive biases is important to ensure better sample efficiency and robustness.
In this study, we propose a novel method of constructing a portfolio based on predicting the distribution of a financial quantity called residual factors, which is known to be generally useful for hedging the risk exposure to common market factors.
The key technical ingredients are twofold.
First, we introduce a computationally efficient extraction method for the residual information, which can be easily combined with various prediction algorithms.
Second, we propose a novel neural network architecture that allows us to incorporate widely acknowledged financial inductive biases such as amplitude invariance and time-scale invariance.
We demonstrate the efficacy of our method on U.S.~and Japanese stock market data. 
Through ablation experiments, we also verify that each individual technique contributes to improving the performance of trading strategies.
We anticipate our techniques may have wide applications in various financial problems.
\end{abstract}

\section{Introduction}\label{sec-introduction}

Developing a profitable trading strategy is a central problem in the financial industry. 
Over the past decade, machine learning and deep learning techniques have driven significant advances across many application areas \cite{Devlin2019BERT,Graves2013RNN}, which inspired investors and financial institutions to develop machine learning-aided trading strategies \cite{wang2019alphastock,choudhry2008hybrid,shah2007machine}. However, it is believed that forecasting the nature of financial time series is essentially a difficult task \cite{krauss2017deep}. 
In particular, the well-known \textit{efficient market hypothesis} \cite{malkiel1970efficient} claims that no single trading strategy could be permanently profitable because the dynamics of the market changes quickly. 
At this point, the financial market is significantly different from stationary environments typically assumed by most machine learning/deep learning methods.

Generally speaking, a good way for deep learning methods to adapt quickly to the given environment is to introduce a network architecture that reflects a good inductive bias for the environment. 
The most prominent examples for such architectures include convolutional neural networks (CNNs) \cite{Krizhevsky2012CNN} for image data and long short-term memories (LSTMs) \cite{Hochreiter1997LSTM} for general time series data. 
Therefore, a natural question to ask is what architecture is effective for processing financial time series.

In finance, researchers have proposed various trading strategies and empirically studied their effectiveness. Hence, it is reasonable to seek architectures inspired by empirical findings in financial studies. In particular, we consider the following three features in the stock market.

\subsection{Hedging exposures to common market factors}
Many empirical studies on stock returns are described in terms of \textit{factor models} (e.g., \cite{fama1992, fama2015}).
These factor models express the return of a certain stock $i$ at time $t$ as a linear combination of $K$ factors plus a residual term:
\begin{equation}
    r_{i, t} = \sum_{k=1}^K \beta^{(k)}_i f^{(k)}_t + \epsilon_{i, t}.
    \label{eq:factor_intro}
\end{equation}
Here, $f^{(1)}_t, \ldots, f^{(K)}_t$ are the common factors shared by multiple stocks $i \in \{ 1, \ldots, S \}$, and the residual factor $\epsilon_{i, t}$ is specific to each individual stock $i$.
Therefore, the common factors correspond to the dynamics of the entire stock market or industries, whereas the residual factors convey some firm-specific information.

In general, if the return of an asset has a strong correlation to the market factors, the asset exhibits a large exposure to the risk of the market. For example, it is known that a classical strategy based on the \textit{momentum} phenomenon \cite{jegadeesh1993moment} is correlated to the Fama--French factors \cite{fama1992, fama2015}, which exhibited negative returns around the credit crisis of 2008 \cite{calomiris2010crisis,szado2009vix}. On the other hand, researchers found that trading strategies based only on the residual factors can be robustly profitable because such strategies can hedge out the time-varying risk exposure to the market factor \cite{blitz2011residual,blitz2013short}. Therefore, to develop trading strategies that are robust to structural changes in the market, it is reasonable to consider strategies based on residual factors.

A natural way to remove the market effect and extract the residual factors is to leverage linear decomposition methods such as principal component analysis (PCA) and factor analysis (FA). In the context of training deep learning-based strategies, it is expected to be difficult to learn such decomposition structures only from the observed data. One possible reason is that the learned strategies can be biased toward using information about market factors because the effect of market factors are dominant in many stocks \cite{pasini2017principal}. Hence, in order to utilize the residual factors, it is reasonable to implement a decomposition-like structure explicitly within the architecture.

\subsection{Designing architectures for scale invariant time series}
When we address a certain prediction task using a neural network-based approach, an effective choice of neural network architecture typically hinges on \textit{patterns} or \textit{invariances} in the data. For example, CNNs \cite{LeCun1999CNN} take into account the shift-invariant structure that commonly appears in image-like data. From this perspective, it is important to find invariant structures that are useful for processing financial time series.

As candidates of such structures, there are two types of invariances known in financial literature. First, it is known that a phenomenon called volatility clustering is commonly observed in financial time series \cite{lux2000volatility}, which suggests an invariance structure of a sequence with respect to its volatility (i.e., amplitude). Second, there is a hypothesis that sequences of stock prices have a certain time-scale invariance property known as the \textit{fractal structure} \cite{peters1994fractal}. We hypothesize that incorporating such invariances into the network architecture is effective at accelerating learning from financial time series data.

\subsection{Constructing portfolios via distributional prediction}

Another important problem is how to convert a given prediction of the returns into an actual trading strategy.
In finance, there are several well-known trading strategies. To name a few, the momentum phenomenon \cite{jegadeesh1993moment} suggests a strategy that bets the current market trend, while the mean reversion \cite{poterba1988reversion} suggests another strategy that assumes that the stock returns moves toward the opposite side of the current direction. However, as suggested by the construction, the momentum and the reversal strategies are negatively correlated to each other, and it is generally unclear which strategy is effective for a particular market.
On the other hand, modern portfolio theory \cite{markowitz1952portfolio} provides a framework to determine a portfolio from distributional properties of asset prices (typically means and variances of returns). The resulting portfolio is unique in the sense that it has an optimal trade-off of returns and risks under some predefined conditions. From this perspective, distributional prediction of returns can be useful to construct trading strategies that can automatically adapt to the market.

\subsection{Summary of contributions}

\begin{figure*}[t]
  \centering
  \includegraphics[width=0.85\linewidth]{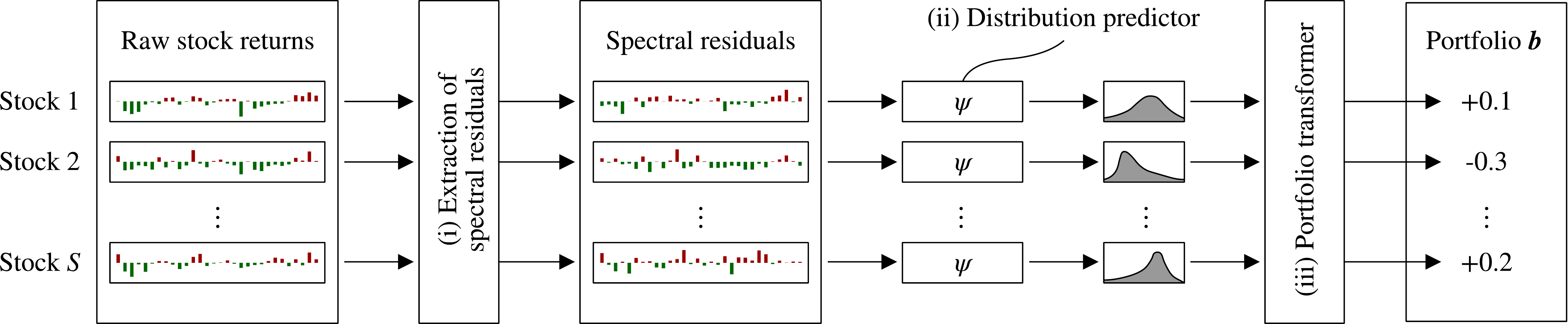}
  \caption{Overview of the proposed system.
  Our system consists of three parts: (i) the extraction layer of residual factors (ii) a neural network-based distribution predictor and (iii) transformation to the optimal portfolio.
  }
  \label{fig-overview}
\end{figure*}

\begin{itemize}
    \item We propose a novel method to extract residual information, which we call the \textit{spectral residuals}. The spectral residuals can be calculated much faster than the classical factor analysis-based method without losing the ability to hedge out exposure to the market factors.
    %Since our method outputs transformed signals with the same dimension as raw return sequences, it can easily be combined with any prediction algorithms.
    Moreover, the spectral residuals can easily be combined with any prediction algorithms.
    \item We propose a new system for distributional prediction of stock prices based on deep neural networks. Our system involves two novel neural network architectures inspired by well-known invariance hypotheses on financial time series. Predicting the distributional information of returns allows us to utilize the optimal portfolio criteria offered by modern portfolio theory.
    \item We demonstrate the effectiveness of our proposed methods on real market data.
\end{itemize}

In the supplementary material, we also include appendices which contain detailed mathematical formulations and experimental settings, theoretical analysis, and additional experiments. 

\section{Preliminaries}\label{sec-preliminary}
%In this section, we provide some basic financial notions and the precise mathematical definitions of our problem setting.

\subsection{Problem setting}
Our problem is to construct a time-dependent portfolio based on sequential observations of stock prices. Suppose that there are $S$ stocks indexed by symbol $i$. The observations are given as a discrete time series of stock prices
$
    \bm{p}^{(i)} = (
    p^{(i)}_1,
    p^{(i)}_2,
    \ldots,
    p^{(i)}_t,
    \ldots
    ).
$
Here, $p^{(i)}_t$ is the price of stock $i$ at time $t$. We mainly consider the \textit{return} of stocks instead of their raw prices. The return of stock $i$ at time $t$ is defined as
$
    r^{(i)}_t = p^{(i)}_{t+1} / p^{(i)}_{t} - 1.
$

A \textit{portfolio} is a (time-dependent) vector of weights over the stocks
$
    \bm{b}_t = (b^{(1)}_t, \ldots, b^{(i)}_t, \ldots, b^{(S)}_t),
$
where $b^{(i)}_t$ is the volume of the investment on stock $i$ at time $t$ satisfying $\sum^S_{i=1} |b^{(i)}_t|=1$. A portfolio $\bm{b}_t$ is understood as a particular trading strategy, that is, $b^{(i)}_t > 0$ implies that the investor takes a long position on stock $i$ with amount $|b^{(i)}_t|$ at time $t$, and $b^{(i)}_t < 0$ means a short position on the stock. Given a portfolio $\bm{b}_t$, its overall return $R_t$ at time $t$ is given as
%\begin{equation}
%\label{eqn-overall-return}
$
    R_t := \sum^S_{i=1} b^{(i)}_t r^{(i)}_t.
$
Then, given the past observations of individual stock returns, our task is to determine the value of $\bm{b}_t$ that optimizes the future returns.

An important class of portfolios is the \textit{zero-investment portfolio} defined as follows.
\begin{definition}[Zero-Investment Portfolio]
A zero-investment portfolio is a portfolio whose buying position and selling position are evenly balanced, i.e., $\sum^S_{i=1} b^{(i)}_t = 0$.
\end{definition}

In this paper, we restrict our attention to trading strategies that output zero-investment portfolio.
This assumption is sensible because a zero-investment portfolio requires no equity and thus encourages a fair comparison between different strategies. 

In practice, there can be delays between the observations of the returns and the actual execution of the trading.
To account for this delay, we also adopt the delay parameter $d$ in our experiments. When we trade with a $d$-day delay, the overall return should be $R_t := R_t^d = \sum^S_{i=1} b^{(i)}_t r^{(i)}_{t+d}$.

\subsection{Concepts of portfolio optimization}
\label{subsection-portfolio-theory}

According to modern portfolio theory \cite{markowitz1952portfolio}, investors construct portfolios to maximize expected return under a specified level of acceptable risk.
The standard deviation is commonly used to quantify the risk or variability of investment outcomes, which measures the degree to which a stock's annual return deviates from its long-term historical average \cite{kintzel2007portfolio}.

The Sharpe ratio \cite{sharpe1994sharpe} is one of the most referenced risk/return measures in finance.
It is the average return earned in excess of the risk-free rate per unit of volatility. The Sharpe ratio is calculated as $(R_\textrm{p} - R_\textrm{f})/\sigma_\textrm{p}$, where $R_\textrm{p}$ is the return of portfolio,
$\sigma_\textrm{p}$ is the standard deviation of the portfolio's excess return, and
$R_\textrm{f}$ is the return of a risk-free asset (e.g., a government bond).
For a zero-investment portfolio, we can always omit $R_\textrm{f}$ since it requires no equity \cite{mitra2009optimal}.

In this paper, we adopt the Sharpe ratio as the objective for our portfolio construction problem.
Since we cannot always obtain an estimate of the total risk beforehand, we often consider sequential maximization of the Sharpe ratio of the next period.
Once we predict the mean vector and the covariance matrix of the population of future returns, the optimal portfolio $\bm{b}^*$ can be solved as
$\bm{b}^* = \lambda^{-1} \bm{\Sigma}^{-1} \bm{\mu}$, where $\lambda$ is a predefined parameter representing the relative risk aversion, $\bm{\Sigma}$ is the estimated covariance matrix, and $\bm{\mu}$ is the estimated mean vector \cite{kan2007optimal}.
\footnote{Note that $\bm{b}^*$ is derived as the maximizer of $\bm{b}^\top \bm{\mu} - \frac{\lambda}{2} \bm{b}^\top \bm{\Sigma} \bm{b}$, where $\bm{b}^\top \bm{\mu}$ and $\bm{b}^\top \bm{\Sigma} \bm{b}$ are the return and the risk of the portfolio $\bm{b}$, respectively.}
Therefore, predicting the mean and the covariance is essential to construct risk averse portfolios.

\section{Proposed System}\label{sec-proposed}
In this section, we present the details of our proposed system, which is outlined in Figure \ref{fig-overview}.
Our system consists of three parts.
In the first part (i), the system extracts residual information to hedge out the effects of common market factors. To this end, in Section \ref{subsection-extraction-residual-factors}, we introduce the \textit{spectral residual}, a novel method based on spectral decomposition.
In the second part (ii), the system predicts future distributions of the spectral residuals using a neural network-based predictor. In the third part (iii), the predicted distributional information is leveraged for constructing optimal portfolios. We will outline these procedures in Section \ref{subsection-distributional-prediction}.
Additionally, we introduce a novel network architecture that incorporates well-known financial inductive biases, which we explain in Section \ref{subsection-architecture}.

\subsection{Extracting residual factors}
\label{subsection-extraction-residual-factors}

As mentioned in the introduction, we focus on developing trading strategy based on the residual factors, i.e., the information remaining after hedging out the common market factors.
Here, we introduce a novel method to extract the residual information, which we call the \textit{spectral residual}.

\subsubsection{Definition of the spectral residuals}
% 1. principal portfolio
First, we introduce some notions from portfolio theory. Let $\bm{r}$ be a random vector with zero mean and covariance $\bm{\Sigma} \in \RR^{S \times S}$, which represents the returns of $S$ stocks over the given investment horizon. Since $\bm{\Sigma}$ is symmetric,
we have a decomposition $\bm{\Sigma} = \bm{V} \bm{\Lambda} \bm{V}^\top$, where $\bm{V} = [\bm{v}_1, \ldots, \bm{v}_S]$ is an orthogonal matrix and $\bm{\Lambda} = \diag(\lambda_1, \ldots, \lambda_S)$ is a diagonal matrix of the eigenvalues. Then, we can create a new random vector as $\hat{\bm{r}} = \bm{V}^\top \bm{r}$ such that the coordinate variables $\hat{r}_i = \bm{v}_i^\top \bm{r}$ are mutually uncorrelated. %$\EE[\hat{r}_i \hat{r}_j] = 0$.
In portfolio theory, $\hat{r}_i$s are called \textit{principal portfolios} \cite{partovi2004principal}. Principal portfolios have been utilized in ``risk parity'' approaches to diversify the exposures to the intrinsic source of risk in the market \cite{meucci2009diversification}.

% 2. spectral residual
Since the volatility (i.e., the standard deviation) of the $i$-th principal portfolio is given as $\sqrt{\lambda_i}$, the raw return sequence has large exposure to the principal portfolios with large eigenvalues. For example, the first principal portfolio can be seen as the factor that corresponds to the overall market \cite{meucci2009diversification}. Therefore, to hedge out common market factors, a natural idea is to discard several principal portfolios with largest eigenvalues. Formally, we define the spectral residuals as follows.

\begin{definition}
Let $C$ ($< S$) be a given positive integer. We define the spectral residual $\tilde{\bm{\epsilon}}$ as a vector obtained by projecting the raw return vector $\bm{r}$ onto the space spanned by the principal portfolios with the smallest $S - C$ eigenvalues.
\end{definition}

In practice, we calculate the empirical version of the spectral residuals as follows.
Given a time window $H > 1$, we define a windowed signal $\bm{X}_t$ as
$
    \bm{X}_t
    := [\bm{r}_{t-H}, \ldots, \bm{r}_{t-1}].
$
We also denote by $\tilde{\bm{X}}_t$ a matrix obtained by subtracting empirical means of row vectors from $\bm{X}_t$. By the singular value decomposition (SVD), $\tilde{\bm{X}}_t$ can be decomposed as
\[
    \tilde{\bm{X}}_t
    = \bm{V}_t
    \underbrace{\diag(\sigma_1, \ldots, \sigma_S) \bm{U}_t}_{\text{principal portfolios}},
\]
where $\bm{V}_t$ is an $S \times S$ orthogonal matrix, $\bm{U}_t$ is an $S \times H$ matrix whose rows are mutually orthogonal unit vectors, and $\sigma_1 \geq \cdots \geq \sigma_S$ are singular values. Note that $\bm{V}_t^\top \tilde{\bm{X}}_t$ can be seen as the realized returns of the principal portfolios. Then, the (empirical) spectral residual at time $s$ ($t-H \le s \le t-1$) is computed as
\begin{equation}
    \tilde{\bm{\epsilon}}_s
    := \bm{A}_t \bm{r}_{s},
    \label{eq:spectral_residual}
\end{equation}
where $\bm{A}_t$ is the projection matrix defined as
\[
    \bm{A}_t := \bm{V}_t
    \mathrm{diag}(\underbrace{0, \ldots, 0}_{C},
    \underbrace{1, \ldots, 1}_{S - C}) \bm{V}_t^\top.
\]
%is the projection matrix onto the space of the bottom $S - C$ principal portfolios.

\subsubsection{Relationship to factor models}
Although we defined the spectral residuals through PCA, they are also related to the generative model \eqref{eq:factor_intro}, and thus convey information about ``residual factors'' in the original sense.

In the finance literature, it has been pointed out that trading strategies depending only on the residual factor $\bm{\epsilon}_t$ in \eqref{eq:factor_intro} can be robust to structural changes in the overall market \cite{blitz2011residual, blitz2013short}. While estimating parameters in the linear factor model \eqref{eq:factor_intro} is typically done by factor analysis (FA) methods \cite{Bartholomew2011},
%\footnote{
%In general, the FA and the PCA result in different solutions unless the residual factor $\bm{\epsilon}_t = (\epsilon_{1, t}, \ldots, \epsilon_{S, t})^\top$ is drawn from a known isotropic distribution.
%},
the spectral residual $\tilde{\bm{\epsilon}}_t$ is not exactly the same as the residual factors obtained from the FA.
Despite this, we can show the following result that the spectral residuals can hedge out the common market factors under a suitable condition.

\begin{prop}\label{prop:uncorrelated}
Let $\bm{r}$ be a random vector in $\RR^S$ generated according to a linear model
$
    \bm{r} = \bm{B} \bm{f} + \bm{\epsilon},
$
where $\bm{B}$ is an $S \times C$ matrix, and $\bm{f} \in \RR^{C}$ and $\bm{\epsilon} \in \RR^{S}$ are zero-mean random vectors.
Assume that the following conditions hold:
\begin{itemize}
    \item $\Var(f_i) = 1$ and $\Var(\epsilon_k) = \sigma > 0$.
    \item The coordinate variables in $\bm{f}$ and $\bm{\epsilon}$ are uncorrelated, that is, $\EE[f_i f_j] = 0$, $\EE[\epsilon_k \epsilon_\ell] = 0$, and $\EE[f_i \epsilon_k]  = 0$ hold for any $i \neq j$ and $k \neq \ell$.
\end{itemize}
Then, we have the followings.
\begin{enumerate}[label=(\roman*)]
    \item The spectral residual $\tilde{\bm{\epsilon}}$ defined in \eqref{eq:spectral_residual} is uncorrelated from the common factor $\bm{f}$.
    \item The covariance matrix of $\tilde{\bm{\epsilon}}$ is given as $\sigma^2 \bm{A}_\mathrm{res}$. Under a suitable assumption\footnote{See \com{Appendix B} for a precise statement.}, this can be approximated as a diagonal matrix, which means the coordinates variables of the spectral residual $\bm{\epsilon}_i$ $(i \in \set{1, \ldots, S})$ are almost uncorrelated.
\end{enumerate}
\end{prop}

%In the above proposition, the assumption on the singular values intuitively means that the common factors have larger volatility contributions than the residual factors.
In the above proposition, the first statement (i) claims that the spectral residual can eliminate the common factors without knowing the exact residual factors.
The latter statement (ii) justifies the diagonal approximation of the predicted covariance, which will be utilized in the next subsection.
For completeness, we provide the proof in \com{Appendix B}. Moreover, the assumption that $\bm{\epsilon}$ is isotropic can be relaxed in the following sense. If we assume that the residual factor $\bm{\epsilon}$ is ``almost isotropic'' and the common factors $\bm{B}\bm{f}$ have larger volatility contributions than $\bm{\epsilon}$, we can show that the linear transformation used in the spectral residual is close to the projection matrix eliminating the market factors.
Since the formal statement is somewhat involved, we give the details in \com{Appendix B}.

Besides, the spectral residual can be computed significantly faster than the FA-based methods. This is because the FA typically requires iterative executions of the SVD to solve a non-convex optimization problem, while the spectral residual requires it only once.
Section \ref{subsection-experiment-1} gives an experimental comparison of running times.

\subsection{Distributional prediction and portfolio construction}\label{subsection-distributional-prediction}

Our next goal is to construct a portfolio based on the extracted information.
To this end, we here introduce a method to forecast future distributions of the spectral residuals, and explain how we can convert the distributional features into executable portfolios.

\subsubsection{Distributional prediction}

Given a sequence of past realizations of residual factors, $\tilde{\epsilon}_{i, t - H}, \ldots, \tilde{\epsilon}_{i, t - 1}$, consider the problem of predicting the distribution of a future observation $\tilde{\epsilon}_{i, t}$.
Our approach is to learn a functional predictor for the conditional distribution $p(\tilde{\epsilon}_{i, t} \mid \tilde{\epsilon}_{i, t - H}, \ldots, \tilde{\epsilon}_{i, t - 1})$.
Since our final goal is to construct the portfolio, we only use predicted means and covariances, and we do not need the full information about the conditional distribution. Despite this, fitting symmetric models such as Gaussian distributions can be problematic because it is known that the distributions of returns are often skewed \cite{cont2000, lin2018}. 
To circumvent this, we utilize quantile regression \cite{koenker_2005}, an off-the-shelf nonparametric method to estimate conditional quantiles. Intuitively, if we obtain a sufficiently large number of quantiles of the target variable, we can reconstruct any distributional properties of that variable.
We train a function $\psi$ that predicts several conditional quantile values, and convert its output into estimators of conditional means and variances.
The overall procedure can be made to be differentiable, so we can incorporate it into modern deep learning frameworks. %See \com{Appendix B.1} for details.

%\subsection{Distributional prediction}
%We provide the details for the distributional prediction introduced in Section 3.2.
%Recall that the task is to estimate the conditional distribution $p(\tilde{\epsilon}_{i, t} \mid \tilde{\epsilon}_{i, t - H}, \ldots, \tilde{\epsilon}_{i, 1})$ from sequential observations of spectral residuals $\tilde{\epsilon}_{i, t}$. Here, $i \in \set{1, \ldots, S}$ is the index representing the stocks and $t$ is the time index.

%Our approach is based on the quantile regression \cite{koenker_2005}, which is a popular nonparametric method to estimate conditional quantiles.

%Let us first consider the problem of predicting the conditional $\alpha$-quantile based on the past observations for some $\alpha \in (0, 1)$. For notation simplicity, we write $y_t = \tilde{\epsilon}_{i, t}$ for the target variable and $\bm{x}_t = (\tilde{\epsilon}_{i, t - H}, \ldots, \tilde{\epsilon}_{i, t - 1})$ for the past $H$ observations.
Here, we provide the details of the aforementioned procedure.
First, we give an overview for the quantile regression objective. Let $Y$ be a scalar-valued random variable, and $\bm{X}$ be another random variable. For $\alpha \in (0, 1)$, an $\alpha$-th conditional quantile of $Y$ given $\bm{X} = \bm{x}$ is defined as
\[
    y(\bm{x}; \alpha) := \inf \set{y': P(Y \leq y' \mid \bm{X} = \bm{x}) \geq \alpha}.
\]
It is known that $y(\bm{x}; \alpha)$ can be found by solving the following minimization problem
\[
    y(\bm{x}; \alpha) = \argmin_{y' \in \RR} \EE[\ell_\alpha(Y, y') \mid \bm{X} = \bm{x}],
\]
where $\ell_\alpha(y, y')$ is the pinball loss defined as
\[
    \ell_\alpha(y, y') := \max \{ (\alpha - 1) (y - y'), \alpha (y - y')\}.
\]

For our situation, the target variable is $y_t = \tilde{\epsilon}_{i, t}$ and the explanatory variable is $\bm{x}_t = (\tilde{\epsilon}_{i, t - H}, \ldots, \tilde{\epsilon}_{i, t - 1})^\top$.
We want to construct a function $\psi: \mathbb{R}^H \to \mathbb{R}$ that estimate the conditional $\alpha$-quantile of $y_t$. To this end, the quantile regression tries to solve the following minimization problem
\[
    \min_{\psi} \widehat{\EE}_{y_{t}, \bm{x}_{t}} [\ell_\alpha(y_{t}, \psi(\bm{x}_{t}))].
\]
Here, $\widehat{\EE}_{y_t, \bm{x}_t}$ is understood as taking the empirical expectation with respect to $y_t$ and $\bm{x}_t$ across $t$.
We should note that a similar application of the quantile regression to forecasting conditional quantiles of time series has been considered in \cite{Biau2011quantile}.

Next, let $Q > 0$ be a given integer, and let $\alpha_j = j/Q$ ($j = 1, \ldots, Q-1$) be an equispaced grid of quantiles. We consider the problem of simultaneously estimating $\alpha_j$-quantiles by a function $\psi: \RR^H \to \RR^{Q - 1}$. To do this, we define a loss function as
\[
    \mcL_Q(y_t, \psi(\bm{x}_t))
    := \sum_{j=1}^{Q - 1} \ell_{\alpha_j}(y_t, \psi_j(\bm{x}_t)),
\]
where $\psi_j(\bm{x}_t)$ is the $j$-th coordinate of $\psi(\bm{x}_t)$.

Once we obtain the estimated $Q - 1$ quantiles $\tilde{y}^{(j)}_t = \psi_j(\bm{x}_t)$ ($j = 1, \ldots, Q-1$), we can estimate the future mean of the target variable $y_t$ as 
\begin{equation}
    \hat{\mu}_t
    := \hat{\mu}(\tilde{\bm{y}}_t)
    = \frac{1}{Q-1} \sum^{Q-1}_{j=1} \tilde{y}_t^{(j)}.
    \label{eq:mu_t}
\end{equation}
Similarly, we can estimate the future variance by the sample variance of $\tilde{y}^{(j)}_t$
\begin{equation}
    \hat{\sigma}_t
    := \hat{\sigma}(\tilde{\bm{y}}_t)
    %\mathrm{Variance}(\tilde{\bm{y}}_t)
    = \frac{1}{Q - 1} \sum (\tilde{y}_t^{(j)} - \hat{\mu}_{t})^2
    \label{eq:sigma_t}
\end{equation}
or its robust counterpart such as the median absolute deviation (MAD).

\subsubsection{Portfolio construction}

Given the estimated means and variances of future spectral residuals, we finally construct a portfolio based on optimality criteria offered by modern portfolio theory \cite{markowitz1952portfolio}.
As mentioned in Section \ref{subsection-portfolio-theory}, the formula for the optimal portfolio requires the means and the covariances of the returns.
Thanks to Proposition \ref{prop:uncorrelated}-(ii), we can
approximate the covariance matrix of the spectral residual by a diagonal matrix.
Precisely, once we calculate the predicted mean $\hat{\mu}_{t, j}$ and the variance $\hat{\sigma}_{t, j}^2$ of the spectral residual at time $t$, the weight for $j$-th asset is given as
$
    \hat{b}_j := \lambda^{-1} \hat{\mu}_{t, j} / \hat{\sigma}_{t, j}^2.
$
%\[
%    \hat{b}_j := \frac{\hat{\mu}_{t, j}}{\lambda \hat{\sigma}_{t, j}^2}.
%\]

In the experiments in Section \ref{sec-experiments}, we compare the performances of zero-investment portfolios.
For trading strategies that do not output zero-investment portfolios, we apply a common transformation to portfolios to be centered and normalized.
As a result, the eventual portfolio does not depend on the risk aversion parameter $\lambda$. See \com{Appendix A.1} for details.

\subsection{Network architectures}\label{subsection-architecture}

% model psi is constructed based on neural network
For the model of the quantile predictor $\psi$, we introduce two architectures for neural network models that take into account scale invariances studied in finance.

\subsubsection{Volatility invariance}
First, we consider an invariance property on amplitudes of financial time series. It is known that financial time series data exhibit a property called volatility clustering \cite{mandelbrot1997variation}. Roughly speaking, volatility clustering describes a phenomenon that large changes in financial time series tend to be followed by large changes, while small changes tend to be followed by small changes. As a result, if we could observe a certain signal as a financial time series, a signal obtained by positive scalar multiplication can be regarded as another plausible realization of a financial time series.

To incorporate such an amplitude invariance property into the model architectures, we leverage the class of positive homogeneous functions. Here, a function $f: \RR^n \to \RR^m$ is said to be positive homogeneous if $f(a \bm{x}) = a f(\bm{x})$ holds for any $\bm{x} \in \RR^n$ and $a > 0$. For example, we can see that any linear functions and any ReLU neural networks with no bias terms are positive homogeneous. More generally, we can model the class of positive homogeneous functions as follows. Let $\tilde{\psi}: S^{H - 1} \to \mathbb{R}^{Q - 1}$ be any function defined on the $H - 1$ dimensional sphere $S^{H - 1} = \{ \bm{x} \in \RR^{H}: \| \bm{x} \| = 1 \}$. Then, we obtain a positive homogeneous function as
\begin{equation}
    \psi(\bm{x})
    = \| \bm{x} \| \tilde{\psi} \left(
        \frac{\bm{x}}{\| \bm{x} \|}
    \right).
    \label{eq:volatility-norm}
\end{equation}
Thus, we can convert any function class on the sphere into the model of amplitude invariant predictors.

\subsubsection{Time-scale invariance}
%\minami{TODO: Compress the explanation}

Second, we consider an invariance property with respect to time-scale. There is a well-known hypothesis that time series of stock prices have \textit{fractal structures} \cite{peters1994fractal}. The fractal structure refers to a self-similarity property of a sequence. That is, if we observe a single sequence in several different sampling rates, we cannot infer the underlying sampling rates from the shape of downsampled sequences.
The fractal structure has been observed in several real markets 
\cite{cao2013chinese, mensi2018islamic, lee2018asymmetric}.
See \com{Remark 1} in \com{Appendix A.2} for further discussion on this property.

To take advantage of the fractal structure, we propose a novel network architecture that we call \textit{fractal networks}. The key idea is that we can effectively exploit the self-similarity by applying a single common operation to multiple sub-sequences with different resolutions. By doing so, we expect that we can increase sample efficiency and reduce the number of parameters to train.
%This idea for utilizing invariance structures is somewhat parallel to well-established approaches in deep learning literature. For example, CNN \cite{LeCun1999CNN} can significantly reduce the number of parameters when applied to shift-invariant data.

%Figure \ref{fig-fractal} illustrates the structure of the fractal network model.
Here, we give a brief overview of the proposed architecture, while a more detailed explanation will be given in \com{Appendix A.2}.
Our model consists of (a) the resampling mechanism and (b) two neural networks $\psi_1$ and $\psi_2$. The input-output relation of our model is described as the following procedure. First, given a single sequence $\bm{x}$ of stock returns, the resampling mechanism $\mathrm{Resample}(\bm{x}, \tau)$ generates a sequence that corresponds to sampling rates specified by a scale parameter $0 < \tau \leq 1$.
We apply $\mathrm{Resample}$ procedure for $L$ different parameters $\tau_1 < \ldots < \tau_L = 1$ and generate $L$ sequences. Next, we apply a common non-linear transformation $\psi_1$ modeled by a neural network. Finally, by taking the empirical mean of these sequences, we aggregate the information on different sampling rates and apply another network $\psi_2$. To sum up, the overall procedure can be written in the following single equation
\begin{equation}
    \psi(\bm{x}) = \psi_2
    \left(
    \frac{1}{L}
    \sum^{L}_{i=1}\psi_1(\mathrm{Resample}(\bm{x}, \tau_i))
    \right).
\end{equation}

\section{Experiments}\label{sec-experiments}
We conducted a series of experiments to demonstrate the effectiveness of our methods on real market data.
In Section \ref{sec-dataset}, we describe the details of the dataset and some common experimental settings used throughout this section.
In Section \ref{subsection-experiment-1}, we test the validity of the spectral residual by a preliminary experiment.
In Section \ref{sec-dpo-experiments}, we evaluate the performance of our proposed system by experiments on U.S.~market data. We also conducted similar experiments on Japanese market data and obtained consistent results. Due to the space limitation, we provide the entire results for Japanese market data in \com{Appendix E}.

\subsection{Dataset description and common settings}\label{sec-dataset}

\subsubsection{U.S.~market data}
%We used daily stock price data from U.S.~market.
For U.S.~market data, we used the daily prices of stocks listed in S\&P 500 from January 2000 to April 2020.
We used data before January 2008 for training and validation and the remainder for testing.
% NOTE: 2000-2005 for training, 2006-2007 for validation, 2008-2020 for evaluation.
We obtained the data from Alpha Vantage \footnote{\url{https://www.alphavantage.co/}}.

We used opening prices because of the following reasons. First, the trading volume at the opening session is larger than that at the closing session \cite{amihud1987trading}, which means that trading with opening prices is practically easier than trading with the closing prices. Moreover, a financial institution cannot trade a large amount of stocks during the closing period because it can be considered as an illegal action known as ``banging the close''.

\subsubsection{Common experimental settings}

%We used some common parameters throughout our experiments.
We adopted the delay parameter $d = 1$ (i.e., one-day delay) for updating portfolios. 
We set the look-back window size as $H = 256$, i.e., all prediction models can access the historical stock prices up-to preceding $256$ business days.
For other parameters used in the experiments, see \com{Appendix C.3} as well.

\subsubsection{Evaluation metrics}
We list the evaluation metrics used throughout our experiments.
\begin{itemize}
    \item \textit{Cumulative Wealth (CW)} is the total return yielded from the trading strategy: $\mathrm{CW}_T := \prod^{T}_{i=1} (1 + R_t)$.
    \item \textit{Annualized Return (AR)} is an annualized return rate defined as
    $\mathrm{AR}_t := (T_\mathrm{Y}/T) \sum^{T}_{i=1} R_t$,
    where $T_\mathrm{Y}$ is the average number of holding periods in a year.
    %\minami{this acronym is confusing with AR(1) model}
    \item \textit{Annualized Volatility (AVOL)} is annualized risk defined as
    $\mathrm{AVOL}_T := ((T_\mathrm{Y}/T) \sum^{T}_{i=1} R_t^2 )^{1/2}$.
    \item \textit{Annualized Sharpe ratio (ASR)} is an annualized risk-adjusted return \cite{sharpe1994sharpe}. It is defined as
    $\mathrm{ASR}_T := \mathrm{AR}_T/\mathrm{AVOL}_T$.
\end{itemize}

As mentioned in Section \ref{subsection-portfolio-theory}, we are mainly interested in ASR as the primary evaluation metric. AR and AVOL are auxiliary metrics for calculating ASR. While CW represents the actual profits, it often ignores the existence of large risk values.
In addition to these, we also calculated some evaluation criteria commonly used in finance: Maximum DrawDown (MDD), Calmar Ratio (CR), and Downside Deviation Ratio (DDR). For completeness, we provide precise definitions in \com{Appendix C.1}.

\subsection{Validity of spectral residuals}\label{subsection-experiment-1}

As suggested in \ref{subsection-extraction-residual-factors}, the spectral residuals can be useful to hedge out the undesirable exposure to the market factors.
To verify this, we compared the performances of trading strategies over (i) the raw returns, (ii) the residual factors extracted by the factor analysis (FA), and (iii) the spectral residuals.

For the FA, we fit the factor model \eqref{eq:factor_intro} with $K = 30$ by the maximum likelihood method \cite{Bartholomew2011} and extracted residual factors as the remaining part. For the spectral residual, we obtain the residual sequence by subtracting $C = 30$ principal components from the raw returns. We applied both methods to windowed data with length $H  = 256$.

In order to be agnostic to the choice of training algorithms, we used a simple reversal strategy. Precisely, for the raw return sequence $\bm{r}_t$, we used a deterministic strategy obtained simply by normalizing the negation of the previous observation $- \bm{r}_{t-1}$ to be a zero-investment portfolio (see \com{Appendix C.2} for the precise formula).
We defined reversal strategies over residual sequences in similar ways.

\begin{figure}[t]
    \centering
    \includegraphics[width=0.9\linewidth]{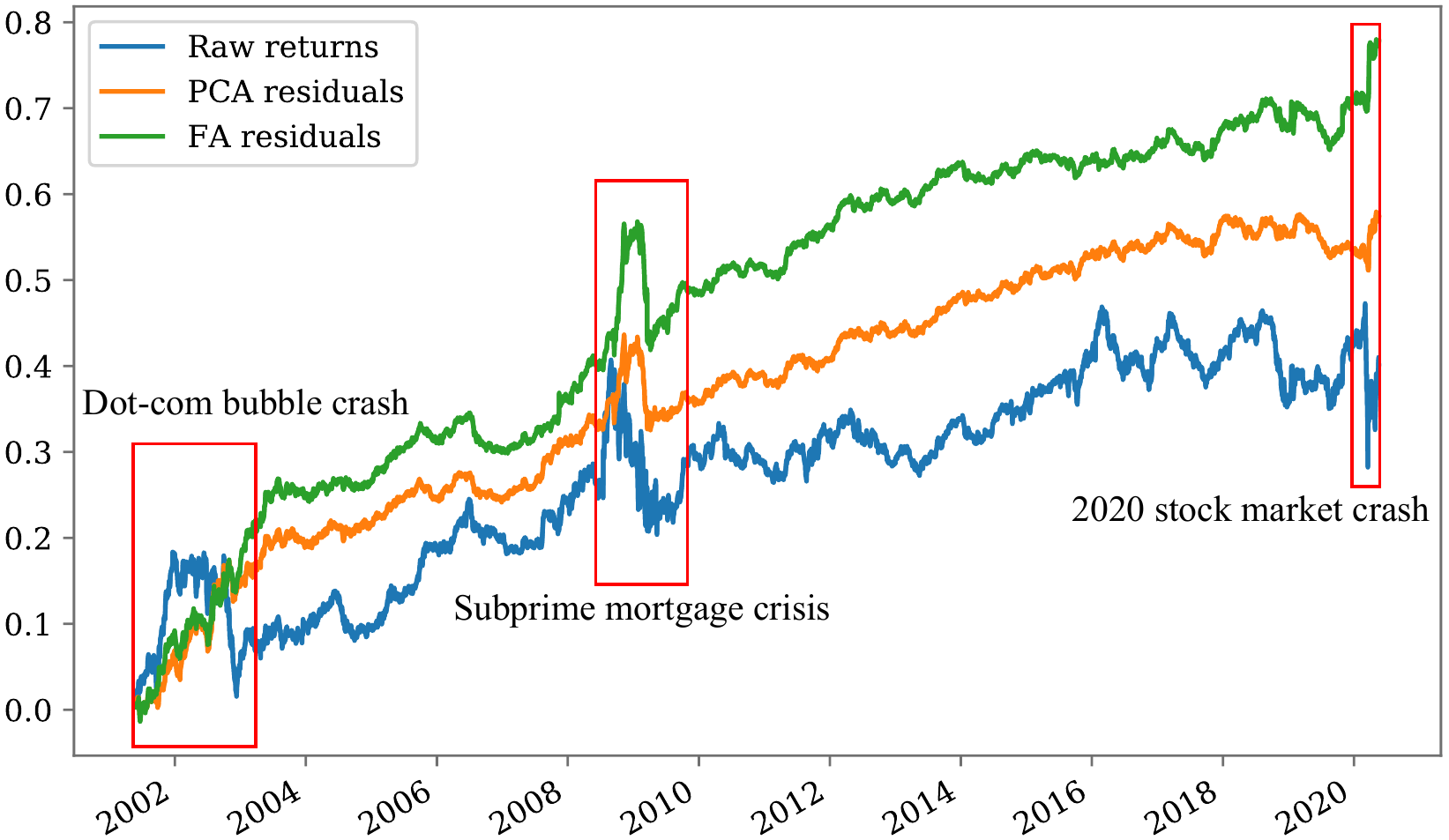}
    \caption{Cumulative returns of reversal strategies over raw returns, the FA residuals, and the spectral residuals. The reversal-based strategies are more robust against financial crises.}
    \label{fig:residual_on_crisis}
\end{figure}

\begin{figure*}[t]
    \begin{minipage}[t]{.48\textwidth}
        \centering
        \begin{tikzpicture}
            \node[inner sep=0] (I) at (0, 0)
            {\includegraphics[width=0.90\textwidth]{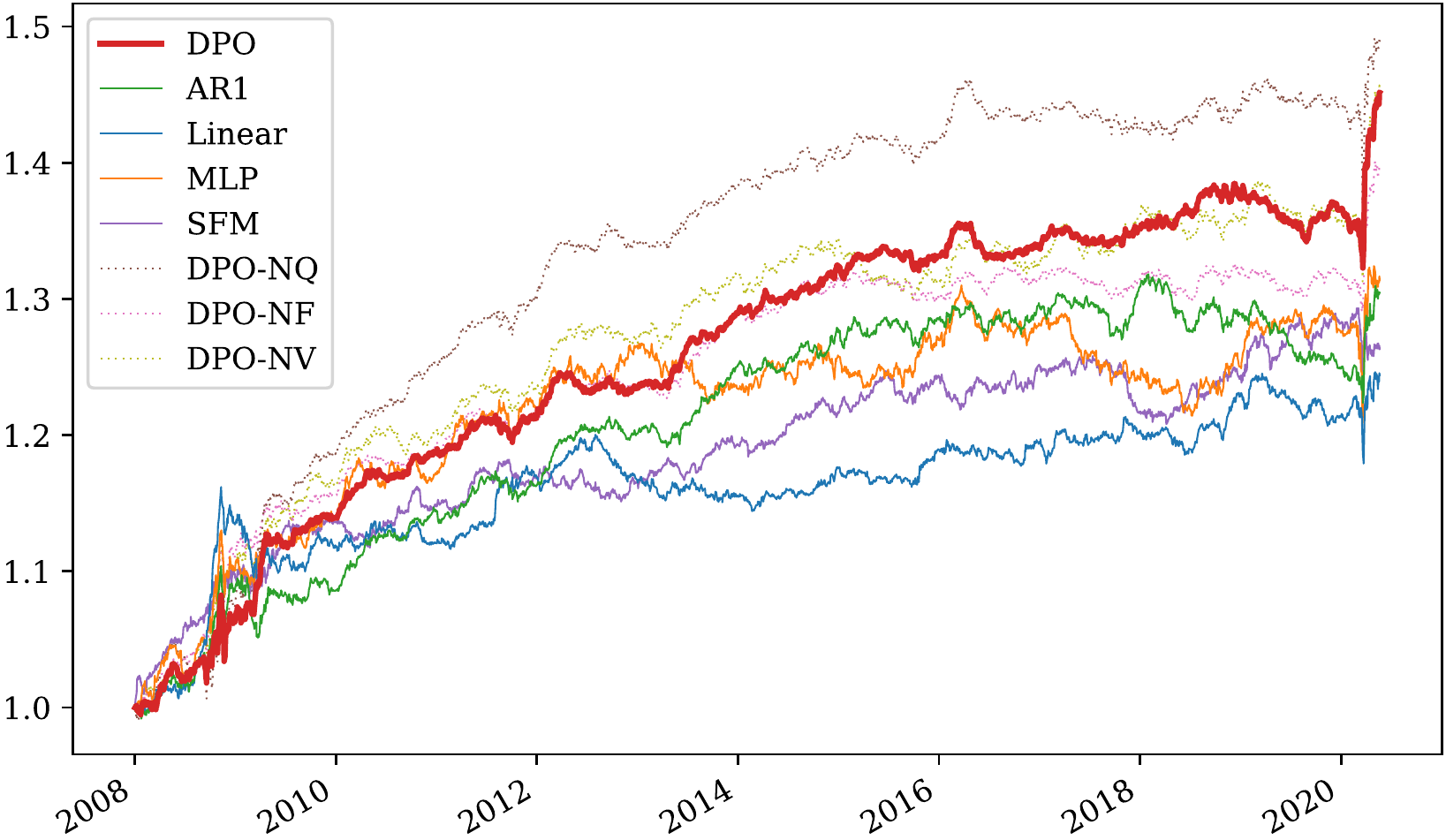}};
            \node[left=0cm of I] {\rotatebox{90}{Cumulative wealth}};
            \node[below=0cm of I] {Date};
        \end{tikzpicture}
        \caption{The Cumulative Wealth in U.S.~market.}
        \label{fig-final-ablation-us}
    \end{minipage}
    \hfill
    \begin{minipage}[t]{.48\textwidth}
        \centering
        \begin{tikzpicture}
            \node[inner sep=0] (I) at (0, 0)
            {\includegraphics[width=0.90\textwidth]{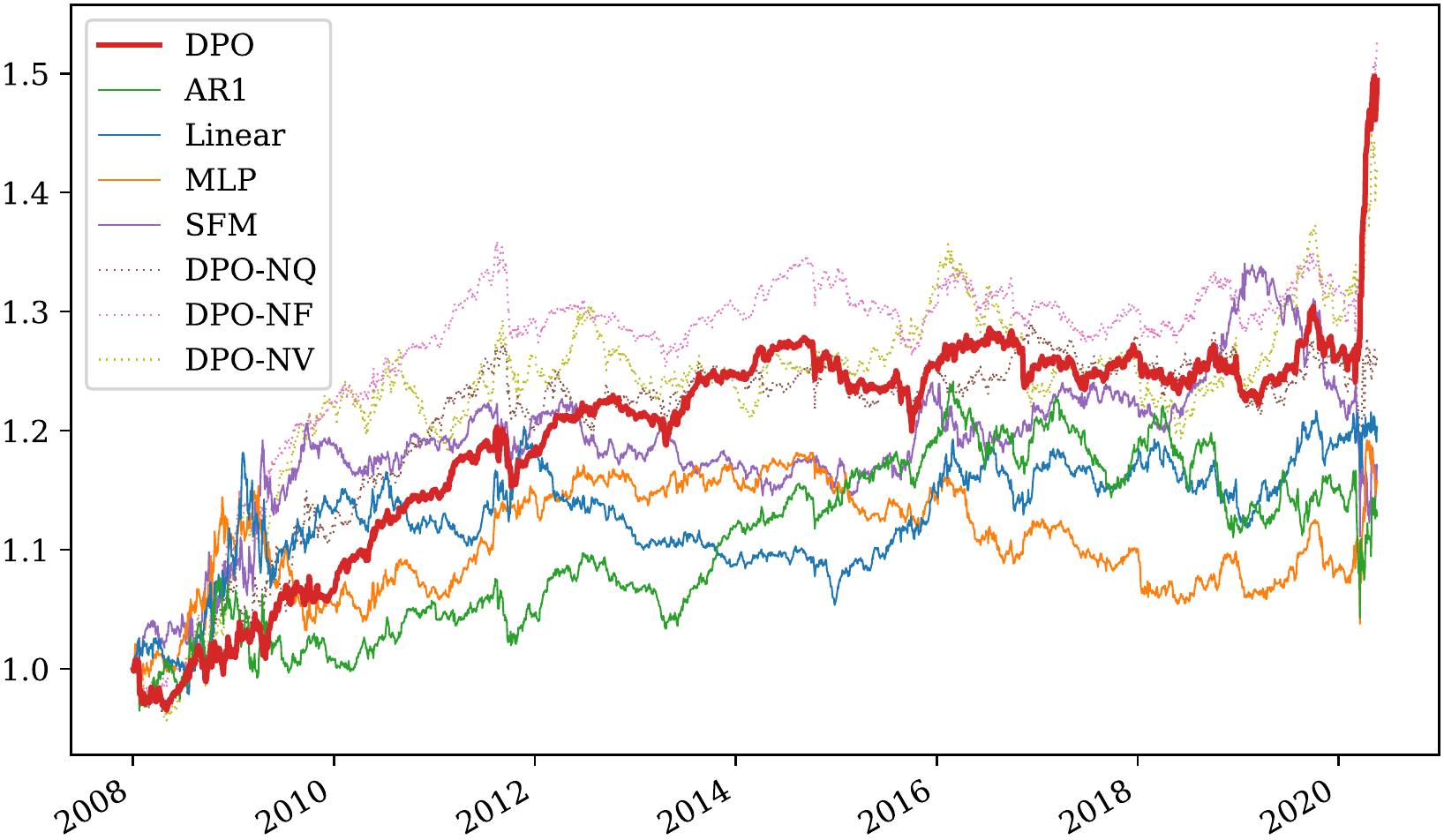}};
            \node[left=0cm of I] {\rotatebox{90}{Cumulative wealth}};
            \node[below=0cm of I] {Date};
        \end{tikzpicture}
        \caption{The Cumulative Wealth in U.S.~market data without spectral residual extraction.
        }
        \label{fig-final-ablation-us-raw}
    \end{minipage}
\end{figure*}

Figure \ref{fig:residual_on_crisis} shows the cumulative returns of reversal strategies performed on the Japanese market data.
We see that the reversal strategy based on the raw returns is significantly affected by several well-known financial crises,
including the dot-com bubble crush in the early 2000s, the 2008 subprime mortgage crisis, and the 2020 stock market crush.
%including 2015-16 Chinese stock market turbulence and 2020 stock market crush.
On the other hand, two residual-based strategies seem to be more robust against these financial crises. The spectral residual performed similarly to the FA residual in cumulative returns. Moreover, in terms of the Sharpe ratio, the spectral residuals ($= 2.86$) performed better than the FA residuals ($= 2.64$).

Remarkably, the spectral residuals were calculated much faster than the FA residuals. In particular, we calculated both residuals using the entire dataset which contains records of all the stock prices for $5,000$ days. For the PCA and the FA, we used implementations in the scikit-learn package \cite{scikit-learn}, and all the computation were run on 18 CPU cores of Intel Xeon Gold 6254 Processor (3.1 GHz). Then, extracting the spectral residuals took approximately 10 minutes, while the FA took approximately 13 hours.

\subsection{Performance evaluation of the proposed system}\label{sec-dpo-experiments}

We evaluated the performance of our proposed system described in Section \ref{sec-proposed} on U.S.~market data. Corresponding results for Japanese market data are provided in \com{Appendix E}.

\subsubsection{Baseline Methods}
We compare our system with the following baselines:
(i) \texttt{Market} is the uniform buy-and-hold strategy.
(ii) \texttt{AR(1)} is the $\mathrm{AR}(1)$ model with all coefficients being $-1$. This can be seen as the simple reversal strategy.
(iii) \texttt{Linear} predicts returns by ordinary linear regression based on previous $H$ raw returns.
(iv) \texttt{MLP} predicts returns by a multi-layer perceptron with batch normalization and dropout \cite{pal1992multilayer,ioffe2015batch,srivastava14dropout}.
(v) \texttt{SFM} is one of state-of-the-art stock price prediction algorithms based on the State Frequency Memory RNNs \cite{zhang2017stock}.

Additionally, we compare our proposed system (\texttt{DPO}) with some ablation models, which are similar to \texttt{DPO} except for the following points.
\begin{itemize}
    \item DPO with No Quantile Prediction (\texttt{DPO-NQ}) does not use the information of the full distributional prediction, but instead it outputs conditional means trained by the $L_2$ loss.
    \item DPO with No Fractal Network (\texttt{DPO-NF}) uses a simple multi-layer perceptron instead of the fractal network.
    \item DPO with No Volatility Normalization (\texttt{DPO-NV}) does not use the normalization \eqref{eq:volatility-norm} in the fractal network.
    %makes its first deep neural network have only one variation for the longest time scale.
\end{itemize}

\subsubsection{Performance on real market data}\label{subsection-experiment-2}

\begin{table}[t]
    \centering
    \caption{Performance comparison on U.S.~market.
    All the methods except for \texttt{Market} are applied to the spectral residuals (SRes).}
    \label{table-final-us}
    {\scriptsize
    \tabcolsep = 5dd
    \begin{tabular}{@{}lcccccc@{}}
    \toprule
    & ASR$\uparrow$ & AR$\uparrow$ & AVOL$\downarrow$ & DDR$\uparrow$ & CR$\uparrow$ & MDD$\downarrow$ \\ \midrule
    \textbf{Market} & +0.607 & \textbf{+0.130} & 0.215 & +0.939 & +0.263 & 0.496 \\ \midrule
    
    \textbf{AR(1) on SRes} & +0.858 & +0.021 & 0.025 & +1.470 & +0.295 & 0.072 \\
    \textbf{Linear on SRes} & +0.724 & +0.017 & 0.024 & +1.262 & +0.298 & 0.059 \\
    \textbf{MLP on SRes} & +0.728 & +0.022 & 0.030 & +1.280 & +0.283 & 0.077 \\
    \textbf{SFM on SRes} & +0.709 & +0.019 & 0.026 & +1.211 & +0.323 & 0.058 \\
    % SFM:  asr: +0.709 ddr: +1.211 cr: +0.323 ar: +0.019 avol: +0.026 mdd: +0.058
    \midrule
    
    \textbf{DPO-NQ} & +1.237 & +0.032 & 0.026 & +2.169 & +0.499 & 0.063 \\
    \textbf{DPO-NF} & +1.284 & +0.027 & \textbf{0.021} & +2.347 & +0.627 & \textbf{0.042} \\
    \textbf{DPO-NV} & +1.154 & +0.030 & 0.026 & +2.105 & +0.562 & 0.053 \\
    \textbf{DPO (Proposed)} & \textbf{+1.393} & +0.030 & \textbf{0.021} & \textbf{+2.561} & \textbf{+0.656} & 0.045 \\
    \bottomrule
    \end{tabular}
    }
\end{table}

% 1. archtectures
Figure \ref{fig-final-ablation-us} shows the cumulative wealth (CW) achieved in U.S.~market. Table \ref{table-final-us} shows the results for the other evaluation metrics presented in Section \ref{sec-dataset}. For parameter $C$ of the spectral residual, we used $C = 10$, which we determined sorely from the training data (see \com{Appendix D.2} for details).
% $C = 50$ for the Japanese market
Overall, our proposed method \texttt{DPO} outperformed the baseline methods in multiple evaluation metrics. Regarding the comparison against three ablation models, we make the following observations.

\begin{enumerate}
    \item \textbf{Effect of the distributional prediction}. We found that introducing the distributional prediction significantly improved the ASR. While \texttt{DPO-NQ} achieved the best CW, \texttt{DPO} performed better in the ASR. It suggests that, without the variance prediction, \texttt{DPO-NQ} tends to pursue the returns without regard to taking the risks. Generally, we observed that \texttt{DPO} reduced the AVOL while not losing the AR. 
    \item \textbf{Effect of the fractal network}. Introducing the fractal network architecture also improved the performance in multiple evaluation metrics. In both markets, we observed that the fractal network contributed to increasing the AR while keeping the AVOL, which is suggestive of the effectiveness of leveraging the financial inductive bias on the return sequence.
    \item \textbf{Effect of the normalization}. We also saw the effectiveness of the normalization \eqref{eq:volatility-norm}. Comparing \texttt{DPO} and \texttt{DPO-NV}, the normalization affected both of the AR and the AVOL, resulting in the improvement in the ASR. This may occur because the normalization improves the sample efficiency by reducing the degrees of freedom of the model.
    %\minami{TODO: introduce the term ``volatility normalization''}
\end{enumerate}

% 2. spectral residuals
\begin{table}[t]
    \centering
    \caption{Performance comparison on U.S.~market without spectral residual extraction.
    %\minami{TODO: change this to Japanese one}
    }
    \label{table-final-us-raw}
    {\scriptsize
    \tabcolsep = 5dd
    
    \begin{tabular}{@{}lcccccc@{}}
    \toprule
    & ASR$\uparrow$ & AR$\uparrow$ & AVOL$\downarrow$ & DDR$\uparrow$ & CR$\uparrow$ & MDD$\downarrow$ \\ \midrule
    %\textbf{Market} & +0.607 & & & +0.939 & +0.263 & +0.496 \\ \midrule
    
    \textbf{AR(1)} & +0.212 & +0.011 & 0.051 & +0.355 & +0.067 & 0.160 \\
    \textbf{Linear} & +0.304 & +0.016 & 0.052 & +0.485 & +0.127 & 0.125 \\
    \textbf{MLP} & +0.261 & +0.013 & 0.048 & +0.424 & +0.103 & 0.122 \\
    % \textbf{LSTNet} & +0.211 & +0.348 & +0.067 & +0.125 \\
    \textbf{SFM} & +0.264 & +0.014 & 0.051 & +0.428 & +0.079 & 0.171 \\
    \midrule
    
    \textbf{DPO-NQ} & +0.405 & +0.020 & 0.048 & +0.655 & +0.172 & 0.114 \\
    \textbf{DPO-NF} & +0.854 & \textbf{+0.034} & 0.040 & +1.472 & +0.436 & 0.078 \\
    \textbf{DPO-NV} & +0.542 & +0.029 & 0.054 & +0.922 & +0.238 & 0.123 \\
    \textbf{DPO} & \textbf{+0.874} & +0.032 & \textbf{0.037} & \textbf{+1.485} & \textbf{+0.524} & \textbf{0.061} \\
    
    \bottomrule
    \end{tabular}
    }
\end{table}

To see the effect of the spectral residuals, we also evaluated our proposed method and the baseline methods on the raw stock returns.
Figure \ref{fig-final-ablation-us-raw} and Table \ref{table-final-us-raw} show the results. Compared to the corresponding results with spectral residuals, we found that the spectral residuals consistently improved the performance for every method. Some further intriguing observations are summarized as follows.

\begin{enumerate}
    \item With the spectral residuals, \texttt{AR(1)} achieved the best ASR among the baseline methods (Table \ref{table-final-us}), which has not been observed on the raw return sequence (Table \ref{table-final-us-raw}). This suggests that the spectral residuals encourage the reversal phenomenon \cite{poterba1988reversion} by suppressing the common market factors.
    Interestingly, without extracting the spectral residuals, the CWs are crossing during the test period, and no single baseline method consistently beats others (Figure \ref{fig-final-ablation-us-raw}). A possible reason is that the strong correlation between the raw stock returns increases the exposure to the common market risks.
    \item We found that our network architectures are still effective on the raw sequences. In particular, \texttt{DPO} outperformed all the other methods in multiple evaluation metrics.
\end{enumerate}

\section{Related Work}

\subsection{Factor Models and Residual Factors}

Trading based on factor models is one of the popular strategies for quantitative portfolio management (e.g., \cite{nakagawa2018deep}).
One of the best-known factor models is Fama and French\cite{fama1992,fama1993common}, %\cite{Fama1970}, 
and they put forward a model explaining returns in the US equity market with three factors: the market return factor, the size (market capitalization) factor and the value (book-to-market) factor.%the market factor, the capitalization size factor and the book value factor.
%A factor model focuses on finding a factor correlated with many stocks such as financial information \cite{harvey2016and}.

Historically, residual factors are treated as errors in factor models\cite{sharpe1964capital}.
However, \cite{blitz2011residual,blitz2013short} suggested that there exists predictability in residual factors.
In modern portfolio theory, less correlation of investment returns enables to earn larger risk-adjusted returns \cite{markowitz1952portfolio}.
Residual factors are less correlated than the raw stock returns by its nature.
Consequently, \cite{blitz2011residual,blitz2013short} demonstrated that residual factors enable to earn larger risk-adjusted returns.

\subsection{Asset Price Prediction using Deep Learning}
With the recent advance of deep learning, various deep neural networks are applied to stock price prediction \cite{chen2019investment}.
Some deep neural networks for time series are also applied to stock price prediction \cite{fischer2018deep}.

Compared to other classical machine learning methods, deep learning enables learning with fewer a priori representational assumptions if provided with sufficient amount of data and computational resources. Even if data is insufficient, introducing inductive biases to a network architecture can still facilitate deep learning \cite{battaglia2018relational}.
Technical indicators are often used for stock prediction (e.g., \cite{metghalchi2012moving,neely2014forecasting}), and \cite{li2019individualized} used technical indicators as inductive biases of a neural network.
\cite{zhang2017stock} used a recurrent model that can analyze frequency domains so as to distinguish trading patterns of various frequencies.

\section{Conclusions}
We proposed a system for constructing portfolios. The key technical ingredients are (i) a spectral decomposition-based method to hedge out common market factors and (ii) a distributional prediction method based on a novel neural network architecture incorporating financial inductive biases. Through empirical evaluations on the real market data, we demonstrated that our proposed method can significantly improve the performance of portfolios on multiple evaluation metrics.
Moreover, we verified that each of our proposed techniques is effective on its own, and we believe that our techniques may have wide applications in various financial problems.

\section*{Acknowledgment}

We thank the anonymous reviewers for their constructive suggestions and comments. We also thank Masaya Abe, Shuhei Noma, Prabhat Nagarajan and Takuya Shimada for helpful discussions.

\bibliography{main}

\begin{thebibliography}{60}
\providecommand{\natexlab}[1]{#1}
\providecommand{\url}[1]{\texttt{#1}}
\providecommand{\urlprefix}{URL }
\expandafter\ifx\csname urlstyle\endcsname\relax
  \providecommand{\doi}[1]{doi:\discretionary{}{}{}#1}\else
  \providecommand{\doi}{doi:\discretionary{}{}{}\begingroup
  \urlstyle{rm}\Url}\fi

\bibitem[{Amihud and Mendelson(1987)}]{amihud1987trading}
Amihud, Y.; and Mendelson, H. 1987.
\newblock Trading mechanisms and stock returns: An empirical investigation.
\newblock \emph{The Journal of Finance} 42(3): 533--553.

\bibitem[{Bartholomew, Knott, and Moustaki(2011)}]{Bartholomew2011}
Bartholomew, D.; Knott, M.; and Moustaki, I. 2011.
\newblock \emph{Latent Variable Models and Factor Analysis: A Unified
  Approach}.
\newblock Wiley, 3rd edition.

\bibitem[{Battaglia et~al.(2018)Battaglia, Hamrick, Bapst, Sanchez-Gonzalez,
  Zambaldi, Malinowski, Tacchetti, Raposo, Santoro, Faulkner
  et~al.}]{battaglia2018relational}
Battaglia, P.~W.; Hamrick, J.~B.; Bapst, V.; Sanchez-Gonzalez, A.; Zambaldi,
  V.; Malinowski, M.; Tacchetti, A.; Raposo, D.; Santoro, A.; Faulkner, R.;
  et~al. 2018.
\newblock Relational inductive biases, deep learning, and graph networks.
\newblock \emph{arXiv preprint arXiv:1806.01261} .

\bibitem[{Biau and Patra(2011)}]{Biau2011quantile}
Biau, G.; and Patra, B. 2011.
\newblock Sequential Quantile Prediction of Time Series.
\newblock \emph{IEEE Transactions on Information Theory} 57(3): 1664--1674.

\bibitem[{Blitz et~al.(2013)Blitz, Huij, Lansdorp, and
  Verbeek}]{blitz2013short}
Blitz, D.; Huij, J.; Lansdorp, S.; and Verbeek, M. 2013.
\newblock Short-term residual reversal.
\newblock \emph{Journal of Financial Markets} 16(3): 477--504.

\bibitem[{Blitz, Huij, and Martens(2011)}]{blitz2011residual}
Blitz, D.; Huij, J.; and Martens, M. 2011.
\newblock Residual momentum.
\newblock \emph{Journal of Empirical Finance} 18(3): 506--521.

\bibitem[{Calomiris, Love, and Peria(2010)}]{calomiris2010crisis}
Calomiris, C.~W.; Love, I.; and Peria, M. S.~M. 2010.
\newblock Crisis" Shock Factors" and the Cross-Section of Global Equity
  Returns.
\newblock Technical report, National Bureau of Economic Research.

\bibitem[{Cao, Cao, and Xu(2013)}]{cao2013chinese}
Cao, G.; Cao, J.; and Xu, L. 2013.
\newblock Asymmetric multifractal scaling behavior in the {Chinese} stock
  market: Based on asymmetric {MF-DFA}.
\newblock \emph{Physica A: Statistical Mechanics and its Applications} 392(4):
  797--807.

\bibitem[{Chen et~al.(2019)Chen, Zhao, Bian, Xing, and
  Liu}]{chen2019investment}
Chen, C.; Zhao, L.; Bian, J.; Xing, C.; and Liu, T.-Y. 2019.
\newblock Investment behaviors can tell what inside: Exploring stock intrinsic
  properties for stock trend prediction.
\newblock In \emph{Proceedings of the 25th ACM SIGKDD International Conference
  on Knowledge Discovery \& Data Mining}, 2376--2384.

\bibitem[{Choudhry and Garg(2008)}]{choudhry2008hybrid}
Choudhry, R.; and Garg, K. 2008.
\newblock A hybrid machine learning system for stock market forecasting.
\newblock \emph{World Academy of Science, Engineering and Technology} 39(3):
  315--318.

\bibitem[{Cont(2000)}]{cont2000}
Cont, R. 2000.
\newblock Empirical properties of asset returns: stylized facts and statistical
  issues.
\newblock \emph{Quantitative Finance} 1: 223--236.

\bibitem[{Davis and Kahan(1970)}]{davis1970}
Davis, C.; and Kahan, W.~M. 1970.
\newblock The Rotation of Eigenvectors by a Perturbation. III.
\newblock \emph{SIAM Journal on Numerical Analysis} 7(1): 1--46.

\bibitem[{Devlin et~al.(2019)Devlin, Chang, Lee, and
  Toutanova}]{Devlin2019BERT}
Devlin, J.; Chang, M.-W.; Lee, K.; and Toutanova, K. 2019.
\newblock BERT: Pre-training of Deep Bidirectional Transformers for Language
  Understanding.
\newblock In \emph{NAACL-HLT}.

\bibitem[{Fama and French(1992)}]{fama1992}
Fama, E.~F.; and French, K.~R. 1992.
\newblock The Cross-Section of Expected Stock Returns.
\newblock \emph{The Journal of Finance} 47(2): 427--465.

\bibitem[{Fama and French(1993)}]{fama1993common}
Fama, E.~F.; and French, K.~R. 1993.
\newblock Common risk factors in the returns on stocks and bonds.
\newblock \emph{Journal of} .

\bibitem[{Fama and French(2015)}]{fama2015}
Fama, E.~F.; and French, K.~R. 2015.
\newblock A five-factor asset pricing model.
\newblock \emph{Journal of Financial Economics} 116(1): 1--22.

\bibitem[{Fischer and Krauss(2018)}]{fischer2018deep}
Fischer, T.; and Krauss, C. 2018.
\newblock Deep learning with long short-term memory networks for financial
  market predictions.
\newblock \emph{European Journal of Operational Research} 270(2): 654--669.

\bibitem[{{Graves}, {Mohamed}, and {Hinton}(2013)}]{Graves2013RNN}
{Graves}, A.; {Mohamed}, A.; and {Hinton}, G. 2013.
\newblock Speech recognition with deep recurrent neural networks.
\newblock In \emph{2013 IEEE International Conference on Acoustics, Speech and
  Signal Processing}, 6645--6649.
\newblock ISSN 2379-190X.

\bibitem[{Grossman and Zhou(1993)}]{grossman1993optimal}
Grossman, S.~J.; and Zhou, Z. 1993.
\newblock Optimal investment strategies for controlling drawdowns.
\newblock \emph{Mathematical finance} 3(3): 241--276.

\bibitem[{Hochreiter and Schmidhuber(1997)}]{Hochreiter1997LSTM}
Hochreiter, S.; and Schmidhuber, J. 1997.
\newblock Long Short-Term Memory.
\newblock \emph{Neural Computation} 9(8): 1735--1780.

\bibitem[{Ioffe and Szegedy(2015)}]{ioffe2015batch}
Ioffe, S.; and Szegedy, C. 2015.
\newblock Batch normalization: Accelerating deep network training by reducing
  internal covariate shift.
\newblock \emph{arXiv preprint arXiv:1502.03167} .

\bibitem[{Jegadeesh and Titman(1993)}]{jegadeesh1993moment}
Jegadeesh, N.; and Titman, S. 1993.
\newblock Returns to Buying Winners and Selling Losers: Implications for Stock
  Market Efficiency.
\newblock \emph{The Journal of Finance} 48(1): 65--91.

\bibitem[{Kan and Zhou(2007)}]{kan2007optimal}
Kan, R.; and Zhou, G. 2007.
\newblock Optimal portfolio choice with parameter uncertainty.
\newblock \emph{Journal of Financial and Quantitative Analysis} 42(3):
  621--656.

\bibitem[{Kingma and Ba(2015)}]{Kingma2015Adam}
Kingma, D.~P.; and Ba, J. 2015.
\newblock Adam: A Method for Stochastic Optimization.
\newblock In \emph{International Conference for Learning Representations
  (ICLR)}.

\bibitem[{Kintzel(2007)}]{kintzel2007portfolio}
Kintzel, D. 2007.
\newblock Portfolio theory, life-cycle investing, and retirement income.
\newblock \emph{Social Security Administration Policy Brief} 2007-02.

\bibitem[{Koenker(2005)}]{koenker_2005}
Koenker, R. 2005.
\newblock \emph{Quantile Regression}.
\newblock Econometric Society Monographs. Cambridge University Press.
\newblock \doi{10.1017/CBO9780511754098}.

\bibitem[{Krauss, Do, and Huck(2017)}]{krauss2017deep}
Krauss, C.; Do, X.~A.; and Huck, N. 2017.
\newblock Deep neural networks, gradient-boosted trees, random forests:
  Statistical arbitrage on the S\&P 500.
\newblock \emph{European Journal of Operational Research} 259(2): 689--702.

\bibitem[{Kristoufek and Vosvrda(2014)}]{kristoufek2014efficiency}
Kristoufek, L.; and Vosvrda, M. 2014.
\newblock Measuring capital market efficiency: long-term memory, fractal
  dimension and approximate entropy.
\newblock \emph{The European Physical Journal B} 87(7): 162.

\bibitem[{Krizhevsky, Sutskever, and Hinton(2012)}]{Krizhevsky2012CNN}
Krizhevsky, A.; Sutskever, I.; and Hinton, G.~E. 2012.
\newblock ImageNet Classification with Deep Convolutional Neural Networks.
\newblock In Pereira, F.; Burges, C. J.~C.; Bottou, L.; and Weinberger, K.~Q.,
  eds., \emph{Advances in Neural Information Processing Systems 25},
  1097--1105. Curran Associates, Inc.
\newblock
  \urlprefix\url{http://papers.nips.cc/paper/4824-imagenet-classification-with-deep-convolutional-neural-networks.pdf}.

\bibitem[{LeCun et~al.(1999)LeCun, Haffner, Bottou, and Bengio}]{LeCun1999CNN}
LeCun, Y.; Haffner, P.; Bottou, L.; and Bengio, Y. 1999.
\newblock Object Recognition with Gradient-Based Learning.
\newblock In \emph{Shape, Contour and Grouping in Computer Vision}, 319--345.
  Springer Berlin Heidelberg.

\bibitem[{Lee et~al.(2018)Lee, Song, Kim, and Chang}]{lee2018asymmetric}
Lee, M.; Song, J.; Kim, S.; and Chang, W. 2018.
\newblock Asymmetric market efficiency using the index-based
  asymmetric-{MFDFA}.
\newblock \emph{Physica A: Statistical Mechanics and its Applications} 512(15):
  1278--1294.

\bibitem[{Li et~al.(2019)Li, Yang, Zhao, Bian, Qin, and
  Liu}]{li2019individualized}
Li, Z.; Yang, D.; Zhao, L.; Bian, J.; Qin, T.; and Liu, T.-Y. 2019.
\newblock Individualized indicator for all: Stock-wise technical indicator
  optimization with stock embedding.
\newblock In \emph{Proceedings of the 25th ACM SIGKDD International Conference
  on Knowledge Discovery \& Data Mining}, 894--902.

\bibitem[{Lin and Liu(2018)}]{lin2018}
Lin, T.-C.; and Liu, X. 2018.
\newblock Skewness, individual investor preference, and the cross-section of
  stock returns.
\newblock \emph{Review of Finance} 22(5): 1841--1876.

\bibitem[{Lux and Marchesi(2000)}]{lux2000volatility}
Lux, T.; and Marchesi, M. 2000.
\newblock Volatility clustering in financial markets: a microsimulation of
  interacting agents.
\newblock \emph{International journal of theoretical and applied finance}
  3(04): 675--702.

\bibitem[{Malkiel and Fama(1970)}]{malkiel1970efficient}
Malkiel, B.~G.; and Fama, E.~F. 1970.
\newblock Efficient capital markets: A review of theory and empirical work.
\newblock \emph{The journal of Finance} 25(2): 383--417.

\bibitem[{Mandelbrot(1997)}]{mandelbrot1997variation}
Mandelbrot, B.~B. 1997.
\newblock The variation of certain speculative prices.
\newblock In \emph{Fractals and scaling in finance}, 371--418. Springer.

\bibitem[{Mandelbrot and Ness(1968)}]{mandelbrot1968fractional}
Mandelbrot, B.~B.; and Ness, J. W.~V. 1968.
\newblock Fractional Brownian Motions, Fractional Noises and Applications.
\newblock \emph{SIAM Review} 10(4): 422--437.

\bibitem[{Markowitz(1952)}]{markowitz1952portfolio}
Markowitz, H. 1952.
\newblock Portfolio selection.
\newblock \emph{The journal of finance} 7(1): 77--91.

\bibitem[{Mensi et~al.(2018)Mensi, Hamdi, Shahzad, Shafiullah, and
  Al-Yahyaee}]{mensi2018islamic}
Mensi, W.; Hamdi, A.; Shahzad, S. J.~H.; Shafiullah, M.; and Al-Yahyaee, K.~H.
  2018.
\newblock Modeling cross-correlations and efficiency of {Islamic} and
  conventional banks from {Saudi Arabia}: Evidence from {MF-DFA} and {MF-DXA}
  approaches.
\newblock \emph{Physica A: Statistical Mechanics and its Applications} 502(15):
  576--589.

\bibitem[{Metghalchi, Marcucci, and Chang(2012)}]{metghalchi2012moving}
Metghalchi, M.; Marcucci, J.; and Chang, Y.-H. 2012.
\newblock Are moving average trading rules profitable? Evidence from the
  European stock markets.
\newblock \emph{Applied Economics} 44(12): 1539--1559.

\bibitem[{Meucci(2009)}]{meucci2009diversification}
Meucci, A. 2009.
\newblock Managing Diversification.
\newblock \emph{Risk} 22(5): 74--79.

\bibitem[{Mitra(2009)}]{mitra2009optimal}
Mitra, S.~K. 2009.
\newblock Optimal combination of trading rules using neural networks.
\newblock \emph{International business research} 2(1): 86--99.

\bibitem[{Nakagawa, Uchida, and Aoshima(2018)}]{nakagawa2018deep}
Nakagawa, K.; Uchida, T.; and Aoshima, T. 2018.
\newblock Deep factor model.
\newblock In \emph{ECML PKDD 2018 Workshops}, 37--50. Springer.

\bibitem[{Neely et~al.(2014)Neely, Rapach, Tu, and Zhou}]{neely2014forecasting}
Neely, C.~J.; Rapach, D.~E.; Tu, J.; and Zhou, G. 2014.
\newblock Forecasting the equity risk premium: the role of technical
  indicators.
\newblock \emph{Management science} 60(7): 1772--1791.

\bibitem[{Pal and Mitra(1992)}]{pal1992multilayer}
Pal, S.~K.; and Mitra, S. 1992.
\newblock Multilayer perceptron, fuzzy sets, classifiaction.
\newblock \emph{IEEE Transactions on Neural Networks} 3(5): 683--697.

\bibitem[{Partovi and Caputo(2004)}]{partovi2004principal}
Partovi, M.~H.; and Caputo, M. 2004.
\newblock Principal portfolios: recasting the efficient frontier.
\newblock \emph{Economics Bulletin} 7(3): 1--10.

\bibitem[{Pasini(2017)}]{pasini2017principal}
Pasini, G. 2017.
\newblock Principal component analysis for stock portfolio management.
\newblock \emph{International Journal of Pure and Applied Mathematics} 115(1):
  153--167.

\bibitem[{Pedregosa et~al.(2011)Pedregosa, Varoquaux, Gramfort, Michel,
  Thirion, Grisel, Blondel, Prettenhofer, Weiss, Dubourg, Vanderplas, Passos,
  Cournapeau, Brucher, Perrot, and Duchesnay}]{scikit-learn}
Pedregosa, F.; Varoquaux, G.; Gramfort, A.; Michel, V.; Thirion, B.; Grisel,
  O.; Blondel, M.; Prettenhofer, P.; Weiss, R.; Dubourg, V.; Vanderplas, J.;
  Passos, A.; Cournapeau, D.; Brucher, M.; Perrot, M.; and Duchesnay, E. 2011.
\newblock Scikit-learn: Machine Learning in {P}ython.
\newblock \emph{Journal of Machine Learning Research} 12: 2825--2830.

\bibitem[{Peters(1994)}]{peters1994fractal}
Peters, E.~E. 1994.
\newblock \emph{Fractal market analysis: applying chaos theory to investment
  and economics}, volume~24.
\newblock John Wiley \& Sons.

\bibitem[{Poterba and Summers(1988)}]{poterba1988reversion}
Poterba, J.~M.; and Summers, L.~H. 1988.
\newblock Mean reversion in stock prices: Evidence and Implications.
\newblock \emph{Journal of Financial Economics} 22(1): 27--59.

\bibitem[{Shah(2007)}]{shah2007machine}
Shah, V.~H. 2007.
\newblock Machine learning techniques for stock prediction.
\newblock \emph{Foundations of Machine Learning| Spring} 1(1): 6--12.

\bibitem[{Sharpe(1964)}]{sharpe1964capital}
Sharpe, W.~F. 1964.
\newblock Capital asset prices: A theory of market equilibrium under conditions
  of risk.
\newblock \emph{The journal of finance} 19(3): 425--442.

\bibitem[{Sharpe(1994)}]{sharpe1994sharpe}
Sharpe, W.~F. 1994.
\newblock The sharpe ratio.
\newblock \emph{Journal of portfolio management} 21(1): 49--58.

\bibitem[{Sortino and Price(1994)}]{sortino1994performance}
Sortino, F.~A.; and Price, L.~N. 1994.
\newblock Performance measurement in a downside risk framework.
\newblock \emph{the Journal of Investing} 3(3): 59--64.

\bibitem[{Srivastava et~al.(2014)Srivastava, Hinton, Krizhevsky, Sutskever, and
  Salakhutdinov}]{srivastava14dropout}
Srivastava, N.; Hinton, G.; Krizhevsky, A.; Sutskever, I.; and Salakhutdinov,
  R. 2014.
\newblock Dropout: A Simple Way to Prevent Neural Networks from Overfitting.
\newblock \emph{Journal of Machine Learning Research} 15(56): 1929--1958.

\bibitem[{Szado(2009)}]{szado2009vix}
Szado, E. 2009.
\newblock VIX futures and options: A case study of portfolio diversification
  during the 2008 financial crisis.
\newblock \emph{The Journal of Alternative Investments} 12(2): 68--85.

\bibitem[{Wang et~al.(2019)Wang, Zhang, Tang, Wu, and
  Xiong}]{wang2019alphastock}
Wang, J.; Zhang, Y.; Tang, K.; Wu, J.; and Xiong, Z. 2019.
\newblock AlphaStock: A Buying-Winners-and-Selling-Losers Investment Strategy
  using Interpretable Deep Reinforcement Attention Networks.
\newblock In \emph{Proceedings of the 25th ACM SIGKDD International Conference
  on Knowledge Discovery \& Data Mining}, 1900--1908.

\bibitem[{Young(1991)}]{young1991calmar}
Young, T.~W. 1991.
\newblock Calmar ratio: A smoother tool.
\newblock \emph{Futures} 20(1): 40.

\bibitem[{Yu, Wang, and Samworth(2014)}]{Yu2014daviskahan}
Yu, Y.; Wang, T.; and Samworth, R.~J. 2014.
\newblock {A useful variant of the Davis–Kahan theorem for statisticians}.
\newblock \emph{Biometrika} 102(2): 315--323.

\bibitem[{Zhang, Aggarwal, and Qi(2017)}]{zhang2017stock}
Zhang, L.; Aggarwal, C.; and Qi, G.-J. 2017.
\newblock Stock price prediction via discovering multi-frequency trading
  patterns.
\newblock In \emph{Proceedings of the 23rd ACM SIGKDD international conference
  on knowledge discovery and data mining}, 2141--2149.

\end{thebibliography}

% for arxiv version
\newpage
\appendix
\section*{Appendix}

This is the supplementary material for the paper entitled ``Deep Portfolio Optimization via Distributional Prediction of Residual Factors''.

\section{Technical details of the proposed method}\label{sec:technical_details}

In this section, we provide some technical details for our proposed method in Section 3.

\subsection{Detailed calculation of portfolio}\label{sec:detailed-portfolio-calculation}
Let $\bm{\mu}_t$ be the expected return vector at time $t$, and $\bm{\Sigma}_t$ be the covariance matrix of returns. As we explained in Section 2.2, the (ideal) optimal portfolio can be written as
$
    \bm{b}_t^* = \lambda^{-1} \bm{\Sigma}_t^{-1} \bm{\mu}_t,
$
where $\lambda > 0$ is a predefined risk aversion parameter. Therefore, having estimators $\hat{\bm{\mu}}_t$ and $\hat{\bm{\Sigma}}_t$ for these population quantities, a natural construction of the portfolio is the following plug-in rule
\[
    \hat{\bm{b}}_t := \frac{1}{\lambda} \hat{\bm{\Sigma}}_t^{-1} \hat{\bm{\mu}}_t.
\]

In our proposed method, we construct a portfolio over the spectral residual $\tilde{\epsilon}_t$. For the estimators of the mean and the covariance, we use the quantile regression-based estimators explained in the previous subsection. In particular, we approximate the covariance matrix by a diagonal matrix $\diag(\hat{\sigma}_{t, 1}, \ldots, \hat{\sigma}_{t, S})$. See also Appendix \ref{sec:orthogonality} for empirical and theoretical justifications of this diagonal approximation. As a result, the weight for $j$-th residual factor is given as
\[
    b_{t, j}^{\mathrm{res}} := \frac{\hat{\mu}_{t, j}}{\lambda \hat{\sigma}_{t, j}},
\]
where $\hat{\mu}_{t, j}$ and $\hat{\sigma}_{t, j}$ are predicted means and variances given as \eqref{eq:mu_t} and \eqref{eq:sigma_t}, respectively.

We now have a portfolio $\bm{b}^\mathrm{res}_t$ defined for the spectral residuals.
Recall that the spectral residual is obtained as a linear transformation of the (centered) raw returns as $\tilde{\bm{\epsilon}}_t = \bm{A}_t \bm{r}_t$ (see Section 3.1).
To obtain a portfolio for the raw return $\bm{r}_t$, we use the relation $\bm{b}_t = \bm{A}_t^\top \bm{b}^\mathrm{res}_t$.

%Finally, we convert $\bm{b}_t$ into a zero-investment portfolio. To this end, we subtract the average $\bar{b}_t = \frac{1}{S} \sum_{i=1}^S b_{t, j}$ from every coordinate and then renormalize the portfolio so that the sum of the absolute values is unity.
In our experiment, we apply a common transformation to the output of any trading strategy so that it becomes a zero-investment portfolio. Here, we explain the detail of this transformation. Let $\bm{b}_t$ be a given portfolio. Assume that $\bm{b}_t$ is not proportional to the all-one vector $\bm{1} = (1, \ldots, 1)^\top$. This also implies that $\bm{b}_t$ is not proportional to the uniform buy-and-hold strategy. To convert $\bm{b}_t$ to a zero-investment portfolio, we subtract the average $\bar{b}_t = \frac{1}{S} \sum_{i=1}^S b_{t, j}$ from every coordinate, and then renormalize the portfolio so that the sum of the absolute values is unity. The resulting portfolio is given as
\[
    \frac{\bm{b}_t - \bar{b}_t \bm{1}}{\norm{\bm{b}_t - \bar{b}_t \bm{1}}_1}.
\]
Note that the normalized portfolio defined in this way does not depend on the parameter $\lambda$.

\subsection{Fractal networks}

Here, we explain the detailed structure of the fractal network introduced in Section 3.3. Figure \ref{fig-fractal} illustrates the structure of the fractal network.

\begin{figure}[t]
  \centering
  \includegraphics[width=0.8\linewidth]{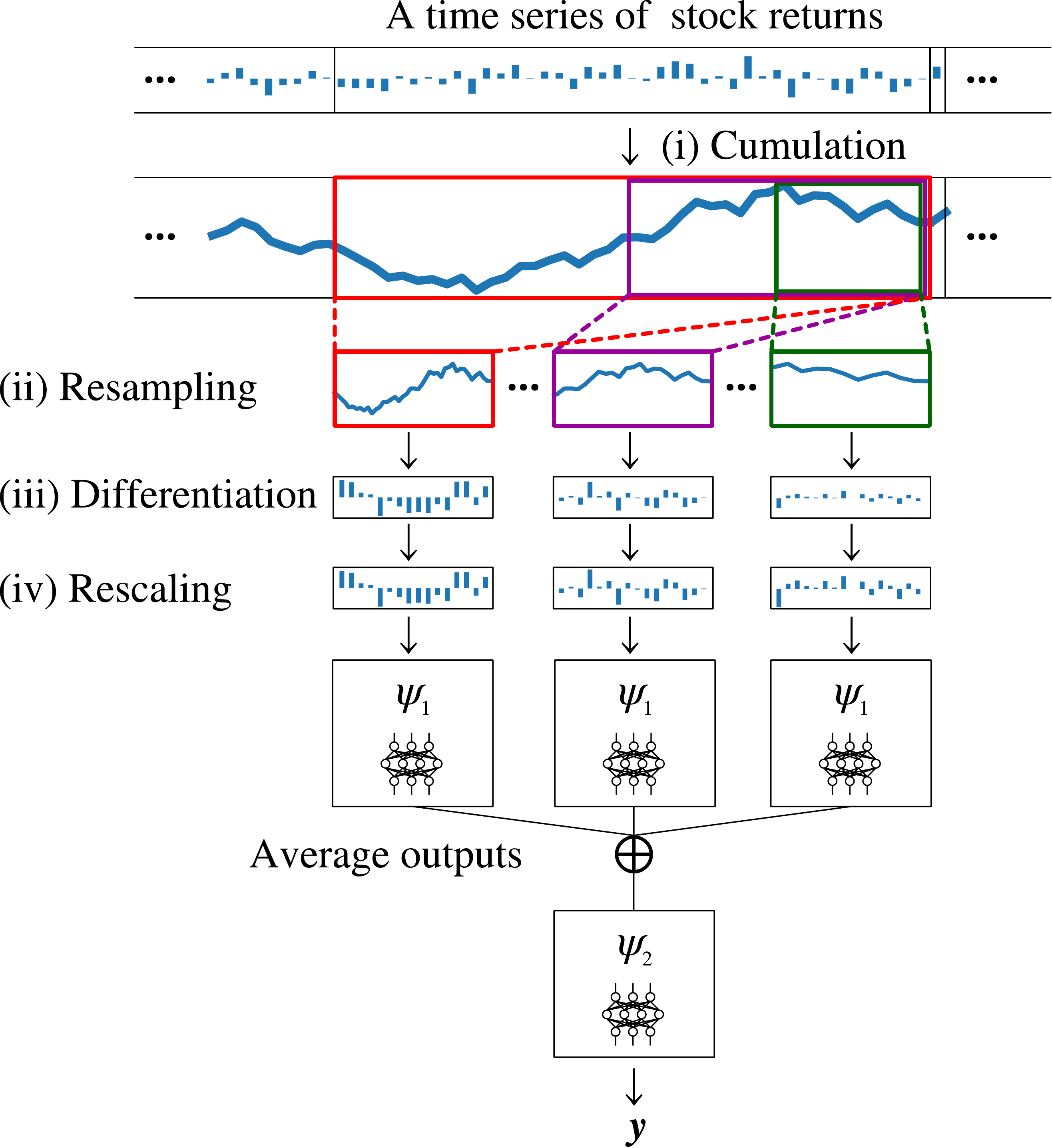}
  \caption{Illustration of the fractal network.}
  \label{fig-fractal}
\end{figure}

Let $\bm{x}$ denote the input of the network. First, the fractal network applies the resampling mechanism $\mathrm{Resample}(\bm{x}, \tau_i)$ for several scale parameters $1 = \tau_0 > \tau_1 > \cdots > \tau_L > 0$ to generate multiple views of $\bm{x}$ with different sampling rates.
To be precise, the resample mechanism outputs a sequence by the following four procedures:
\begin{enumerate}
    \renewcommand{\labelenumi}{(\roman{enumi})}
    \item \textbf{Cumulation.} First, given an input $\bm{x} = (x_1, \ldots, x_H)$, we calculate the cumulative sum $\bm{z} = (x_1, x_1 + x_2, \ldots, x_1 + \cdots + x_H)$. Since the input $\bm{x}$ corresponds to the residual factors of the returns, each coordinate variable $x_s$ is understood as the increment or the difference of stock prices of two adjacent time periods. Hence, we use its cumulative sum so that it corresponds to (logarithmic) stock prices and exhibits the fractal structure.
    \item \textbf{Resampling.} Second, given $0 < \tau \leq 1$ and $\bm{z}$, we generate a shorter sequence $\bm{z}_\tau$ of a fixed length $H' (< H)$ by downsampling from $(z_{\lfloor (1 - \tau) H \rfloor}, \ldots, z_H)$. In other words, the resulting sequence is a subsequence with length $\lceil \tau H \rceil$ located at the end of the original sequence.
    \item \textbf{Differentiation.} Third, we again take the first difference of the downsampled sequence $\bm{z}_\tau$ to obtain the corresponding return sequence.
    \item \textbf{Rescaling.} Finally, we normalize the output by multiplying the entire sequence by $\tau^{-1/2}$. On the choice of the multiplicative factor, see following Remark \ref{rmk:hurst} as well.
\end{enumerate}

Next, we apply a common non-linear transformation $\psi_1$ to these views. Finally, we take average of the results and apply another transformation $\psi_2$. The overall procedure is summarized in the following equation
\begin{equation}
    \psi(\bm{x})
    = \psi_2\left(
        \frac{1}{L} \sum_{i=1}^L
        \psi_1 (\mathrm{Resample}(\bm{x}, \tau_i))
    \right).
    \label{eq:app_fractal_network}
\end{equation}

To incorporate the volatility invariance (Section 3.3) into the fractal network, we want to make the overall transformation $\bm{x} \mapsto \psi(\bm{x})$ positive homegeneous. To do this, it suffices to ensure that non-linear transformations $\psi_1$ and $\psi_2$ are positive homogeneous. This can be proved by combining the following claims.

\begin{lem}
Let $f_1, f_2, \ldots, f_L$ be any positive homogeneous functions.
\begin{enumerate}
    \renewcommand{\labelenumi}{(\roman{enumi})}
    \item Any linear transformation is positive homogeneous.
    \item Suppose the composition $f_1 \circ f_2$ can be defined. Then, $f_1 \circ f_2$ is positive homogeneous.
    \item The concatenation $\bm{x} \mapsto (f_1(\bm{x})^\top, \ldots, f_L(\bm{x})^\top)^\top$ is positive homogeneous.
    \item The average $\bm{x} \mapsto \frac{1}{L} \sum_{i=1}^L f_i(\bm{x})$ is positive homogeneous.
    \item The resampling mechanism $\bm{x} \mapsto \mathrm{Resample}(\bm{x}, \tau)$ is positive homogeneous for any $0 < \tau \leq 1$.
\end{enumerate}
\end{lem}
\begin{proof}
Since (i), (ii), and (iii) are almost obvious, we omit their proofs. (iv) is derived from (i), (ii), and (iii).
As for (v), the resampling mechanism is a linear transformation, and thus it is positive homogeneous. In fact, it is easy to see that the four operations in the resampling mechanism (i.e., cumulation, resampling, differentiation, and rescaling) are linear. Therefore, the resampling mechanism is also linear.
\end{proof}

\begin{rmk}\label{rmk:hurst}
In the rescaling phase, the choice of multiplicative factor $\tau^{-1/2}$ is sensible because of the following reason. If the underlying law is the fractional Brownian motion with the Hurst index $\mcH$, the appropriate scaling factor determined by its self-similarity is $\tau^{-\mcH}$ \cite{mandelbrot1968fractional}. It has been reported in \cite{kristoufek2014efficiency} that, in many real-world markets, estimated Hurst indices are approximately $\mcH \approx 0.5$, which implies that stock prices exhibit similar self-similarity properties as the standard Brownian motion.
\end{rmk}

\begin{rmk}
In some paper, the term ``fractal property'' stands for the behavior of stochastic processes captured by the fractional Brownian motion with $\mcH \neq 1/2$, which can lead inefficiency of the market. In this paper, however, we focus on the self-similarity of processes to design the network architecture. The self-similarity can be observed even in efficient markets with $\mcH \approx 0.5$.
\end{rmk}

\section{Theoretical analysis}\label{sec:theoretical}

In this section, we provide theoretical analyses of the spectral residuals justifying that the spectral residuals can hedge out the market factors. In Section \ref{sec:theory_isotropic}, we prove Proposition 1, which shows that the spectral residuals are actually uncorrelated to the market factors when the true residual factors are isotropic.
In practice, the covariance of the residual factors can be approximated by diagonal matrices. In Section \ref{sec:orthogonality}, we explain why this approximation is justified.
In Section \ref{sec:theory_anisotropic}, we investigate a more general situation where the true residual factors are anisotropic, i.e., the variances of the coordinate variables differ.

Before proceeding, let us recall the definition of the spectral residual. Let $\bm{r}$ be a zero-mean random vector with a covariance matrix $\bm{\Sigma} \in \RR^{S \times S}$. Let $\bm{\Sigma} = \bm{V} \bm{\Lambda} \bm{V}^\top$ is an eigendecomposition of $\bm{\Sigma}$, where $\bm{V} = [\bm{v}_1, \ldots, \bm{v}_S]$ is an orthogonal matriix and $\bm{\Lambda} = \diag(\lambda_1, \ldots, \lambda_S)$ is a diagonal matrix of the eigenvalues. Assume that $\lambda_1, \ldots, \lambda_S$ are sorted in descending order, i.e., $\lambda_1 \geq \cdots \geq \lambda_S$. Given an integer $1 \leq C < S$, let
\[
    \bm{V}_{\mathrm{res}} = [\bm{v}_{C + 1}, \ldots, \bm{v}_S]
\]
be an $S \times (S - C)$ matrix that consists of the eigenvectors that correspond to the smallest $S-C$ eigenvalues\footnote{
    In general, $\bm{V}_{\mathrm{res}}$ is not unique because there can be multiple eigenvectors having the same eigenvalues. For simplicity, we assume that $\lambda_{C} > \lambda_{C + 1}$ so that $\bm{V}_{\mathrm{res}}$ is uniquely determined.
}.
We also define a matrix $\bm{A}_\mathrm{res}$ as
\[
    \bm{A}_\mathrm{res} = \bm{V}_{\mathrm{res}} \bm{V}_{\mathrm{res}}^\top.
\]
Note that $\bm{A}_\mathrm{res}$ is the orthogonal projection onto the space spanned by the principal portfolios with the smallest $S-C$ eigenvalues. Then, we define the spectral residual as
\begin{equation}
    \tilde{\bm{\epsilon}} = \bm{A}_\mathrm{res} \bm{r}.
    \label{eq:spectral_residual_alt}
\end{equation}

\subsection{Isotropic residuals}\label{sec:theory_isotropic}

First, we prove the following result, which corresponds to Proposition \ref{prop:uncorrelated} in the main body. 

\begin{prop}[Proposition 1 in the main body, restated]\label{prop:uncorrelated_appendix}
Let $\bm{r}$ be a random vector in $\RR^S$ generated according to a linear model
$
    \bm{r} = \bm{B} \bm{f} + \bm{\epsilon},
$
where $\bm{B}$ is an $S \times C$ matrix of full column rank, and $\bm{f} \in \RR^{C}$ and $\bm{\epsilon} \in \RR^{S}$ are zero-mean random vectors.
Assume that the following conditions hold:
\begin{itemize}
    \item $\Var(f_i) = 1$ and $\Var(\epsilon_k) = \sigma > 0$.
    \item The coordinate variables in $\bm{f}$ and $\bm{\epsilon}$ are uncorrelated, that is, $\EE[f_i f_j] = 0$, $\EE[\epsilon_k \epsilon_\ell] = 0$, and $\EE[f_i \epsilon_k]  = 0$ hold for any $i \neq j$ and $k \neq \ell$.
    %\item The smallest singular value of $\bm{B}$ is not less than the largest standard deviation of $\epsilon_i$ ($i \in \set{1, \ldots, S}$).
\end{itemize}
Then, we have the followings.
\begin{enumerate}[label=(\roman*)]
    \item The spectral residual $\tilde{\bm{\epsilon}}$ defined in \eqref{eq:spectral_residual_alt} is uncorrelated from the common factor $\bm{f}$.
    \item The covariance matrix of $\tilde{\bm{\epsilon}}$ is given as $\sigma^2 \bm{A}_\mathrm{res}$. Under a suitable assumption, this can be approximated as a diagonal matrix.
\end{enumerate}
\end{prop}

\begin{proof}
%\minami{(ii) is wrong, so I just prove (i). Maji doushiyou}
Let $\bm{B} \bm{B}^\top = \bm{V}_{\mathrm{mar}} \bm{\Lambda} \bm{V}_{\mathrm{mar}}^\top$ be the eigendecomposition of the market factor covariance $\bm{B} \bm{B}^\top$, where $\bm{V}_{\mathrm{mar}} = [\bm{v}_1, \ldots, \bm{v}_S]$ is an orthogonal matrix and $\bm{\Lambda}$ is a diagonal matrix of the eigenvalues. Since $\mathord{\mathrm{rank}}\bm{B} \bm{B}^\top \leq C$, we can write
\[
    \bm{\Lambda} = \diag(\lambda_1, \ldots, \lambda_C, 
    \underbrace{0, \ldots, 0}_{S - C}).
\]
Note that $\lambda_i >  0$ for all $i \in \set{1, \ldots, C}$ since $\bm{B} \bm{B}^\top$ is positively semi-definite and has rank $C$.
Since coordinate variables in $\bm{f}$ and $\bm{\epsilon}$ are uncorrelated, the covariance matrix of $\bm{r}$ is calculated as
\[
    \bm{\Sigma}
    = \bm{B} \bm{B}^\top + \sigma^2 \bm{I}
    = \bm{V}_{\mathrm{mar}}(\bm{\Lambda} + \sigma^2 \bm{I}) \bm{V}_{\mathrm{mar}}^\top.
\]
This gives the eigendecomposition of $\bm{\Sigma}$, and its eigenvalues are given as
\[
    \underbrace{\lambda_1 + \sigma^2, \ldots, \lambda_C + \sigma^2}_{C},
    \underbrace{\sigma^2, \ldots, \sigma^2}_{S-C}.
\]
In particular, the projection matrix for the spectral residual is given as
\begin{equation}
    \bm{A}_{\mathrm{res}} =
    [\bm{v}_{C + 1}, \ldots, \bm{v}_S] [\bm{v}_{C+1}, \ldots, \bm{v}_S]^\top.
    \label{eq:ares_construction}
\end{equation}
Consequently, we have $\tilde{\bm{\epsilon}} =  \bm{A}_{\mathrm{res}} (\bm{B} \bm{f} + \bm{\epsilon}) = \bm{A}_{\mathrm{res}} \bm{\epsilon}$, which is uncorrelated from $\bm{B} \bm{f}$ (thus we proved (i)).

Moreover, the covariance matrix of $\bm{A}_{\mathrm{res}} \bm{\epsilon}$ is given as $\sigma^2 \bm{A}_{\mathrm{res}} \bm{A}_{\mathrm{res}}^\top = \sigma^2 \bm{A}_{\mathrm{res}}$. Regarding the diagonal approximation of this matrix, see the next subsection for details.
\end{proof}

\subsection{On the near-orthogonality of spectral residuals}\label{sec:orthogonality}

When we construct a portfolio over the spectral residuals, we use an diagonal approximation of the covariance matrix (see Section 3.2 and Appendix \ref{sec:detailed-portfolio-calculation}). As mentioned in the previous subsection, this covariance matrix is proportional to $\bm{A}_\mathrm{res}$. Hence, to justify the diagonal approximation, we should show that $\bm{A}_\mathrm{res}$ becomes nearly diagonal.

\begin{figure*}
    \centering
    \includegraphics[width=1.0\linewidth]{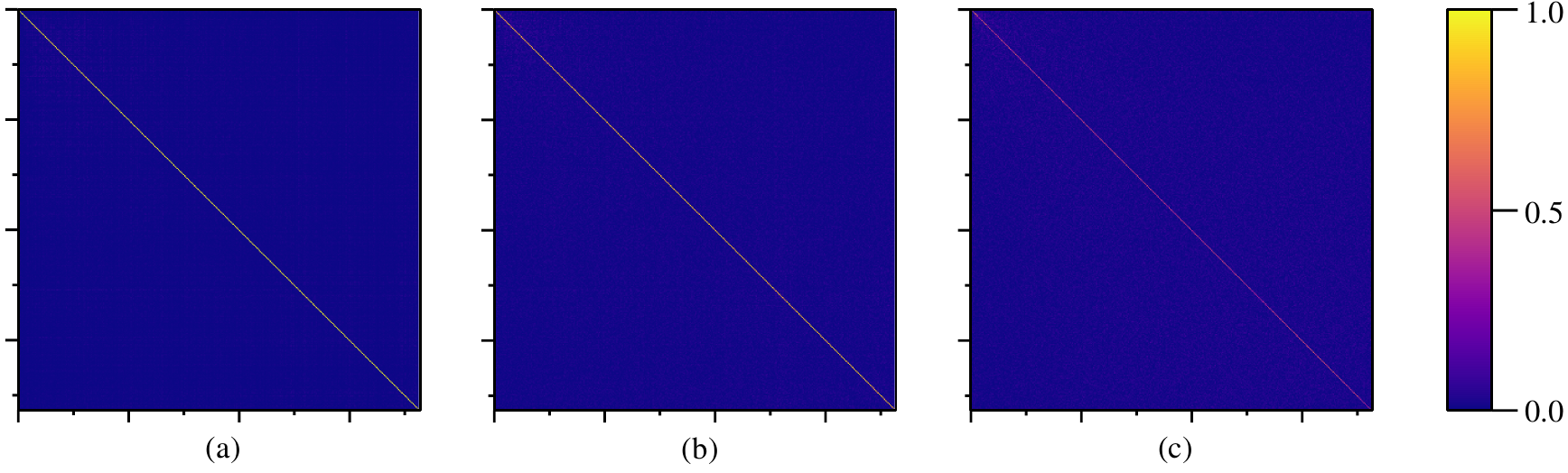}
    \caption{The absolute values of projection matrix $\bm{A}_\mathrm{res}$ calculated from U.S.~market data.
    Three panels show results for (a) $C = 10$, (b) $C = 50$, and (c) $C = 200$, respectively.
    In all cases, the projection matrices are reasonably close to diagonal matrices.
    }
    \label{fig-ac-visualization}
\end{figure*}

Figure \ref{fig-ac-visualization} shows examples of $\bm{A}_\mathrm{res}$ calculated emprically from U.S. market data. For several choices of the parameter $C$, these matrices are reasonably close to diagonal matrices.

But why does this happen? Generally speaking, a projection matrix constructed as \eqref{eq:ares_construction} is not necessarily approximated well by a diagonal matrix. However, in financial markets, we may assume that the following properties hold for the first several principal portfolios, which may justify the diagonal approximation.

\begin{itemize}
    \item \textbf{Well-spreading factors}: In a financial market, a large proportion of a stock return is often explained by a small number of common market factors. A market factor may be almost equally shared by all the stocks in the market. For example, the first principal portfolio can be seen as the factor of the overall market, which may be close to the ``equal weighting index'' of stocks. In this case, removing the market factor changes all the elements of the covariance matrix only very slightly.
    \item \textbf{Spiked factors}: Suppose that, in an investment horizon, a certain company's stock price is affected by a factor that is independent from any other market factors. In this situation, there can be a principal portfolio that arises only from a single stock. Removing this factor is just equivalent to removing the corresponding stock, which does not affect the off-diagonal elements of the covariance matrix.
\end{itemize}

To formulate the above idea, we introduce the following two notions.

\begin{definition}
Let $\bm{v} \in \RR^S$ be a unit vector (i.e., $\norm{\bm{v}}_2 = 1$). Let $\delta > 0$ be a positive number.
\begin{itemize}
    \item We say that $\bm{v}$ is \textit{$\delta$-spreading} if $|v_i| \leq \sqrt{\delta / S}$ for all $i \in \set{1, \ldots, S}$.
    \item We say that $\bm{v}$ is \textit{$\delta$-spiked} if there exists $i^* \in \set{1, \ldots, S}$ such that $|v_i| \leq \delta/S$ for all $i \neq i^*$.
\end{itemize}
\end{definition}

The above properties ensure that the off-diagonal elements of $\bm{v} \bm{v}^\top$ are sufficiently small.

\begin{prop}
Suppose that a unit vector $\bm{v} \in \RR^S$ is either $\delta$-spreading or $\delta$-spiked for some $\delta > 0$. Then, for all $i \neq j$, the $(i, j)$ element of $\bm{P} = \bm{v} \bm{v}^\top$ is bounded as $|P_{ij}| = |v_i v_j| \leq \delta/S$.
\end{prop}
\begin{proof}
The conclusion is obvious if $\bm{v}$ is $\delta$-spreading. If $\bm{v}$ is $\delta$-spiked, we have $|v_i v_j| \leq 1 \cdot \delta / S = \delta/S$ for any $i \neq j$.
\end{proof}

Now, let us consider the linear factor model $\bm{r} = \bm{B} \bm{f} + \bm{\epsilon}$. We assume a similar condition as in Proposition \ref{prop:uncorrelated_appendix}. In particular, we assume that $\bm{\epsilon}$ has an isotropic covariance matrix $\sigma^2 \bm{I}$. Our goal is to show that the spectral residual $\bm{A}_\mathrm{res} \bm{r}$ has a nearly isotropic covariance matrix under a suitable condition.

Let $\bm{B} \bm{B}^\top = \bm{V}_\mathrm{mar} \bm{\Lambda} \bm{V}_\mathrm{mar}^\top$ be an eigendecomposition of the market factor, where $\bm{V}$ is an orthogonal matrix and $\bm{\Lambda} = \diag(\lambda_1, \ldots, \lambda_C, 0, \ldots, 0)$ with $\lambda_1 \geq \cdots \geq \lambda_C > 0$. Then, as in Proposition \ref{prop:uncorrelated_appendix}, we have
\begin{align*}
    \bm{A}_\mathrm{res}
    &= [\bm{v}_{C+1}, \ldots, \bm{v}_S] [\bm{v}_{C+1}, \ldots, \bm{v}_S]^\top \\
    &= \bm{I} - \bm{P}_\mathrm{mar},
\end{align*}
and the spectral residual is given as $\bm{A}_\mathrm{res} \bm{r} = \bm{A}_\mathrm{res} \bm{\epsilon}$.
Here,
\[
    \bm{P}_\mathrm{mar}
    = [\bm{v}_{1}, \ldots, \bm{v}_C] [\bm{v}_{1}, \ldots, \bm{v}_C]^\top
\]
is the projection matrix onto the space of the market factors.
Thus, the covariance matrix of the spectral residual is given as $\sigma^2 \bm{A}_\mathrm{res}$. The following proposition explains that this covariance matrix is nearly diagonal if the top-$C$ principal portfolios are either ``well-spreading'' or ``nearly spiked''.

\begin{prop}
We use similar notations and assumptions as above. Suppose that each $\bm{v}_i$ ($i \in \set{1, \ldots, C}$) is either $\delta$-spreading or $\delta$-spiked. Then, the off-diagonal elements of $\bm{A}_\mathrm{res}$ are bounded as
\[
    |[\bm{A}_\mathrm{res}]_{i,j}| \leq \frac{C \delta}{S}
\]
for any $i \neq j$.
\end{prop}

\begin{proof}
Combining the equality $\bm{P}_\mathrm{mar} = \sum_{i=1}^C \bm{v}_i \bm{v}_i^\top$ and the fact that all the absolute values of the off-diagonal elements in $\bm{v}_i \bm{v}_i^\top$ are not larger than $\delta/S$, we have $|[\bm{A}_\mathrm{res}]_{i,j}| = |[\bm{P}_\mathrm{mar}]_{i,j}| \leq C \delta / S$ for any $i \neq j$. Hence the result. 
\end{proof}

\subsection{Anisotropic residuals}\label{sec:theory_anisotropic}

Propositon \ref{prop:uncorrelated_appendix} does not hold when the true residual factor $\epsilon_1, \ldots, \epsilon_S$ are anisotropic. However, we can show that the spectral residual is almost uncorrelated from the market factors if
\begin{itemize}
    \item the volatilities of market factors are sufficiently larger than the residual factors, and
    \item the residual factors are almost isotropic.
\end{itemize}
Here, we provide a quantitative justification of this claim by using matrix perturbation theory \cite{davis1970, Yu2014daviskahan}.

% TBD
% frobenius norm
% principal angle
Before describing the result, we provide some notations. For any $S \times S$ matrix $\bm{A}$, $\norm{\bm{A}}_\mathrm{F}$ is the Frobenius norm defined as $\norm{\bm{A}}_\mathrm{F}^2 = \sum_{i, j} a_{ij}^2$. For any orthogonal projection matrix $\bm{P}$, let $\mcV(\bm{P})$ denote the corresponding linear subspace (i.e., the largest linear subspace that is invariant under $\bm{P}$).
Let $\bm{P}_1$ and $\bm{P}_2$ be any two orthogonal projection matrices having the same rank, and $\mcV_i = \mcV(\bm{P}_i)$ ($i = 1, 2$). Then, the quantity
\begin{align*}
    \norm{\sin \Theta(\mcV_1, \mcV_2)}_\mathrm{F}
    := \norm{\bm{P}_1 (\bm{I} - \bm{P}_2)}_\mathrm{F}
    = \norm{\bm{P}_2 (\bm{I} - \bm{P}_1)}_\mathrm{F}
\end{align*}
measures the ``principal angles'' between two subspaces $\mcV_1$ and $\mcV_2$ \cite{davis1970}. In particular, this quantity becomes zero if and only if the two subspaces agree.

\begin{prop}
Let $\bm{r}$ be a random vector in $\RR^S$ generated according to a linear model
$
    \bm{r} = \bm{B} \bm{f} + \bm{\epsilon},
$
where $\bm{B}$ is an $S \times C$ matrix of full rank, and $\bm{f} \in \RR^{C}$ and $\bm{\epsilon} \in \RR^{S}$ are zero-mean random vectors.
Assume that the following conditions hold:
\begin{itemize}
    \item $\Var(f_i) = 1$ and $\Var(\epsilon_k) = \sigma_k^2 > 0$ for $k \in \set{1, \ldots, S}$.
    \item The coordinate variables in $\bm{f}$ and $\bm{\epsilon}$ are uncorrelated, that is, $\EE[f_i f_j] = 0$, $\EE[\epsilon_k \epsilon_\ell] = 0$, and $\EE[f_i \epsilon_k]  = 0$ hold for any $i \neq j$ and $k \neq \ell$.
    %\item The smallest singular value of $\bm{B}$ is not less than the largest standard deviation of $\epsilon_i$ ($i \in \set{1, \ldots, S}$).
\end{itemize}
Let $\bm{P}_\mathrm{mar}$ be the orthogonal projection matrix onto the linear space spanned by $\bm{B}$, and let $\bm{A}_\mathrm{res}$ be the projection matrix of the spectral residuals.
Let $\lambda_\mathrm{min}$ be the smallest positive eigenvalue of $\bm{B} \bm{B}^\top$.
Then, we have
\begin{equation}
    \norm{\bm{P}_\mathrm{mar}\bm{A}_\mathrm{res}}_\mathrm{F}
    \leq \frac{2 \sqrt{S}(\max_{i} \sigma_i^2 - \min_i \sigma_i^2)}{\lambda_\mathrm{min}}.
    \label{eq:anisotropic_result}
\end{equation}
\end{prop}

Here, we give an interpretation of \eqref{eq:anisotropic_result}. If $\norm{\bm{P}_\mathrm{mar}\bm{A}_\mathrm{res}}_\mathrm{F}$ is shown to be small, the spectral residual $\bm{A}_\mathrm{res} \mathrm{r}$ is nearly orthogonal to any market factors. The right-hand side of \eqref{eq:anisotropic_result} can be small when either of the following conditions are satisfied: (i) If the residual factors are nearly isotropic, then $\max_{i} \sigma_i^2 - \min_i \sigma_i^2$ is small. (ii) If the smallest variance of the market factor $\lambda_\mathrm{min}$ is much larger than $\sqrt{S} (\max_{i} \sigma_i^2 - \min_i \sigma_i^2)$, the right-hand side becomes small.

\begin{proof}
Since $\bm{f}$ and $\bm{\epsilon}$ are uncorrelated, the covariance matrix of the generative model $\bm{r} = \bm{B} \bm{f} + \bm{\epsilon}$ can be written as $\bm{\Sigma} = \bm{B}\bm{B}^\top + \bm{Q}$, where $\bm{Q} = \diag(\sigma_1^2, \ldots, \sigma_S^2)$ is the variance of the residual factors.

Let $\bar{s} = \frac{1}{S} \sum_{i=1}^S \sigma_i^2$ be the averaged variance of the residual factors. As such, $\bar{s} \bm{I}$ is the closest isotropic covariance matrix of $\bm{Q}$ in the following sense:
\[
    \bar{s} \bm{I} \in \argmin_{s \bm{I}: s \geq 0} \norm{\bm{Q} - s \bm{I}}_{\mathrm{F}}.
\]
Let $\widehat{\bm{Q}} = \bar{s} \bm{I} - \bm{Q}$. Define $ \Delta_\mathrm{iso} \geq 0$ as
\[
    \Delta_\mathrm{iso} := \norm{\widehat{\bm{Q}}}_\mathrm{F}
    = \min_{s \bm{I}: s \geq 0} \norm{\bm{Q} - s \bm{I}}_{\mathrm{F}}
    = \left(
        \sum_{i=1}^S (\sigma_i^2 - \bar{s})^2
    \right)^{1/2},
\]
which quantifies how far $\bm{Q}$ is away from being isotropic.

As in the proof of Proposition \ref{prop:uncorrelated_appendix}, let $\bm{B} \bm{B}^\top = \bm{V}_\mathrm{mar} \bm{\Lambda} \bm{V}_\mathrm{mar}^\top$ be the eigendecomposition of $\bm{B} \bm{B}^\top$.
Define $\widehat{\bm{\Sigma}}$ as
\[
    \widehat{\bm{\Sigma}} = \bm{\Sigma} + \widehat{\bm{Q}}
    = \bm{B}\bm{B}^\top + \bar{s} \bm{I}.
\]
Since $\bm{B}\bm{B}^\top$ and $\bar{s} \bm{I}$ are simultaneously diagonalizable by $\bm{V}_\mathrm{mar}$, we can write
\[
    \widehat{\bm{\Sigma}}
    = \bm{V}_\mathrm{mar}(\bm{\Lambda} + \bar{s} \bm{I}) \bm{V}_\mathrm{mar}^\top.
\]
In particular, the eigenvalues of $\widehat{\bm{\Sigma}}$ are given as
\[
    \underbrace{\lambda_1 + \bar{s}, \ldots, \lambda_C + \bar{s}}_{C},
    \underbrace{\bar{s}, \ldots, \bar{s}}_{S-C}.
\]
Let $\bm{P}_\mathrm{mar}$ be the orthogonal projection matrix onto the linear subspace spanned by the column vectors of $\bm{B}$ (i.e., the market factors). Note that such a linear subspace is spanned by $\bm{v}_1, \ldots, \bm{v}_C$, and $\bm{P}_\mathrm{mar}$ is computed as
\[
    \bm{P}_\mathrm{mar}
    = [\bm{v}_1, \ldots, \bm{v}_C] [\bm{v}_1, \ldots, \bm{v}_C]^\top.
\]

Also, let $\bm{\Sigma} = \bm{U} \bm{M} \bm{U}^\top$ be the eigendecomposition of $\bm{\Sigma}$, where $\bm{U} = [\bm{u}_1, \ldots, \bm{u}_S]$ is an orthogonal matrix, and $\bm{M} = \diag(\mu_1, \ldots, \mu_S)$ is a diagonal matrix with $\mu_1 \geq \cdots \geq \mu_C > \mu_{C  + 1} \geq \cdots \mu_S > 0$
\footnote{
    This is true for many cases where the market factors are much larger than the residual factors. For example, let $\delta = \norm{\hat{\bm{Q}}}_\mathrm{op} = \max_j |\sigma_j^2 - \bar{s}|$, and suppose that $\lambda_C = \min_{i} \lambda_i > 2\delta$. Then, from the well-known Weyl's inequality on eigenvalues, we have $\mu_C \geq \lambda_C + \bar{s} - \delta > \bar{s} + \delta \geq \mu_{C + 1}$.
}.
Then, the projection matrix for the spectral residual is given as
\[
    \bm{A}_\mathrm{res}
    = [\bm{u}_{C + 1}, \ldots, \bm{u}_{S}] [\bm{u}_{C + 1}, \ldots, \bm{u}_{S}]^\top.
\]

Now, $\bm{P}_\mathrm{mar}$ and $\bm{I} - \bm{A}_\mathrm{res}$ are the projection matrices that correspond to eigenvectors with the largest $C$ eigenvalues of $\widehat{\bm{\Sigma}}$ and $\bm{\Sigma}$, respectively. From Davis--Kahan $\sin \theta$ theorem \cite{davis1970, Yu2014daviskahan}, we conclude
\[
    \norm{\bm{P}_\mathrm{mar} \bm{A}_\mathrm{res}}_\mathrm{F}
    \leq \frac{2 \norm{\widehat{\bm{\Sigma}} - \bm{\Sigma}}_\mathrm{F}}{\min_{1 \leq i \leq C} \lambda_i}
    = \frac{2 \Delta_\mathrm{iso}}{\min_{1 \leq i \leq C} \lambda_i}.
\]
We also have a weaker inequality \eqref{eq:anisotropic_result} since $\Delta_\mathrm{iso} \leq \sqrt{S}(\max_i \sigma_i^2 - \min_{i} \sigma_i^2)$.
%\minami{interpretation: (i) If the residual factors are nearly isotropic, then $\Delta_\mathrm{iso}$ is small. (ii) Note that $\Delta_\mathrm{iso} \leq \sqrt{S} (\max_{i} \sigma_i^2 - \min_{i} \sigma_i^2) \leq \sqrt{S} \max_{i} \sigma_\mathrm{max}^2$. Therefore, if the smallest variance of the market factor $\lambda_\mathrm{min}$ is much larger than $\sqrt{S} \sigma_{\mathrm{max}}^2$, $\bm{A}_\mathrm{res}$ vanishes almost all market factors.}

\end{proof}

\section{Details for experimental settings}\label{sec:details_experiments}

In this section, we provide some more details on the experiments in Section 4.

\subsection{Definitions of additional evaluation metrics}

\begin{itemize}
    \item \textit{Maximum DrawDown (MDD)} is the maximum loss from a peak to a trough \cite{grossman1993optimal}, which can measure one aspect of downside risk:
    \begin{equation}
    \mathrm{MDD}_T =
        \max_{t \in \{1,\ldots,T\} }
        \max_{s \in \{1,\ldots,t\} }
        \left(
        \frac{\mathrm{CW}_t - \mathrm{CW}_s}{\mathrm{CW}_t}
        \right).
    \end{equation}
    
    \item \textit{Calmar Ratio (CR)} is a risk-adjusted return based on the maximum drawdown \cite{young1991calmar}:
    $\mathrm{CR}_T := \mathrm{AR}_T / \mathrm{MDD}_T$.
    
    \item \textit{Downside Deviation Ratio (DDR)} (a.k.a.~Sortino Ratio) is a variation of the Sharpe ratio \cite{sortino1994performance}. While the Sharpe ratio regards overall volatility as risk, it regards only volatility caused by negative returns as harmful risk:
    \begin{equation}
        \mathrm{DDR}_T
        := \frac{\mathrm{AR}_T}{\sqrt{
        \frac{T_\mathrm{Y}}{T} \sum^{T}_{i=1} \min(0, R_t)^2
        }}.
    \end{equation}
\end{itemize}

\subsection{Definitions of some baseline methods}

\begin{itemize}
    \item \textbf{Reversal strategy}. In Section 4.2 and Section 4.3, we used a simple reversal strategy (\texttt{AR(1)}) as a benchmark method, which is defined as follows. Let $\bm{r}_t$ ($t = 1, 2, \ldots$) be either a return sequence or a transformed return sequence (e.g., the spectral residuals). Then, the reversal strategy $\bm{b}_t$ is defined by renormalizing $- \bm{r}_{t-1}$ to be a zero-investment portfolio, that is,
    \[
        \bm{b}_t = \frac{\bm{r}_{t-1} - \bar{r}_{t-1} \bm{1}}{\norm{\bm{r}_{t-1} - \bar{r}_{t-1} \bm{1}}_1},
    \]
    where $\bar{r}_{t-1} = \frac{1}{S}\sum_{i=1}^S r_{t-1}$.
    \item \textbf{MLP-based prediction}. In Section 4.2, we also used a neural network-based prediction of returns (\texttt{MLP}). We used a fully-connected neural network with $4$ hidden layers. Each hidden layer has 512 nodes, 50\% dropout and batch normalization.
\end{itemize}

\subsection{Details of network architectures and hyperparameters}

Here, we explain the details of the architecture that we used for the distributional prediction. For each coordinate $i \in \set{1, \ldots, S}$ of the spectral residual, we applied a common non-linear function $\psi: \RR^{H} \to \RR^{Q-1}$ that predicts $Q - 1$ quantiles based on the past $H$ observations. We designed $\psi$ by the fractal network introduced in Section 3.3. The detailed specifications are as follows.

\begin{itemize}
    \item We estimate $Q - 1 = 31$ quantiles for each stock. In particular, $j$-th coordinate of $\psi$ corresponds to the $\frac{j}{32}$-th quantile of the future distribution.
    \item The resampling mechanism $\mathrm{Resample}(\bm{x}, \tau)$ outputs sequences of a fixed length $H' = 64$. For the scale parameters, we used $\tau_j := 4^{-j/20}$ with $j \in \set{0, \ldots, 21}$.
    \item For the function $\psi_1: \RR^{H'} \to \RR^{K}$ (with $K - 256$), we used a fully connected neural network with 3 hidden layers. Each layer has 256 nodes, 50\% dropout and batch normalization.
    \item For the function $\psi_2: \RR^{K} \to \RR^{Q-1}$, we used a fully connected neural network with 8 hidden layers. Each layer has 128 nodes, 50\% dropout and batch normalization.
    \item For training, we used the Adam optimizer \cite{Kingma2015Adam} with learning rate $0.001$.
\end{itemize}

\section{Supplementary experiments for spectral residuals}

In this section, we provide additional experimental results on the spectral residual. In Section \ref{sec:ex-local-stability}, we conduct a simple experiment to show that the short-term behavior of the (empirical) spectral residual is reasonably stable. In Section \ref{sec:ex-choosing-c}, we discuss a simple way to determine the parameter $C$ from the training data.

\subsection{Local stability of spectral residuals}\label{sec:ex-local-stability}

We check whether the spectral residuals are locally stable. This is not trivial because (i) the spectral residual is calculated locally on time windows, and (ii) the long-term behavior of the financial time series is highly non-stationary. To this end, we here investigate the ability of the spectral residuals to reduce the volatility of sequences.

For each time $t$, we calculated the projection matrix ($\bm{A}_t$ defined in Section 3.1) and the spectral residuals $\tilde{\bm{\epsilon}}_s$ for $t - H \leq s \leq t - 1$. We also generated the spectral residuals for unseen duration $t \leq s < t + H$ by fixing $\bm{A}_t$ and extrapolating the relation $\tilde{\bm{\epsilon}}_s = \bm{A}_t \bm{r}_s$. Then, we calculated the volatility $\mathrm{Vol}(\tilde{\bm{\epsilon}}_s)$ of the spectral residual at each $s$. Throughout, we fixed $H = 256$, and we varied the number of principal components to be removed as $C \in \{ 0, 1, 10, 20, 50, 100 \}$.

\begin{figure}[t]
  \centering
    \begin{tikzpicture}
        \node[inner sep=0] (I) at (0, 0)
        {\includegraphics[width=0.7\linewidth]{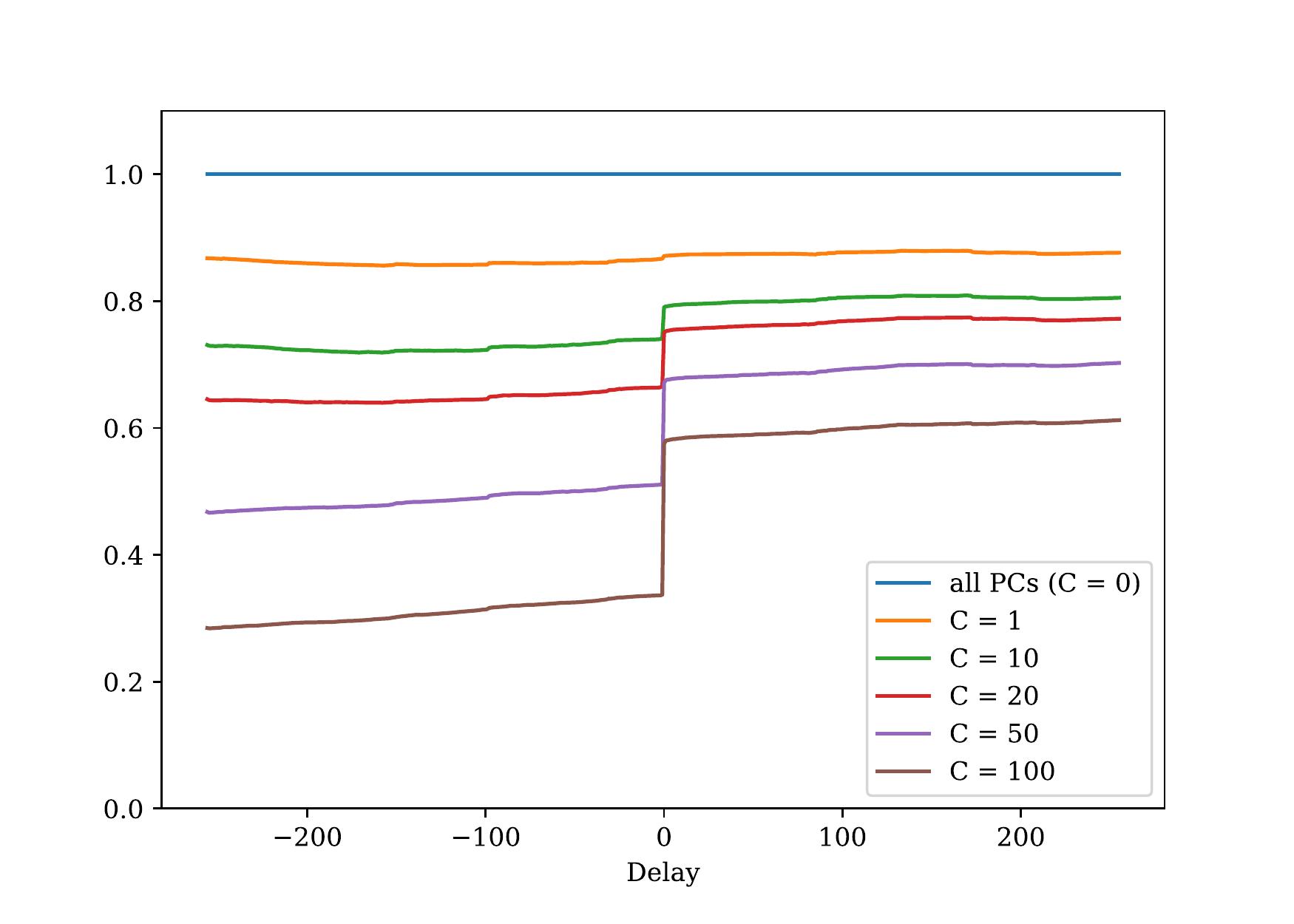}};
        \node[left=0cm of I] {\rotatebox{90}{Relative volatility}};
        \node[below=0cm of I] {Delay $\Delta$};
    \end{tikzpicture}

    \vspace{-10pt}

    \caption{
    Relative volatility to the raw stock returns for various choices of parameter $C$. The ability to reduce the volatility seems to continue in the unseen duration ($\Delta > 0$), which may suggest the local stability of spectral residuals.  See the text for details.
    }
    \label{fig-pca-stationariness}
\end{figure}

Figure \ref{fig-pca-stationariness} shows the result averaged over $t$, which illustrates the proportion of the volatility of spectral residuals ($C \geq 1$) to the volatility of raw stock returns ($C = 0$). The horizontal axis is the delay parameter $\Delta := s - t \in \{ -H, \ldots, 0, \ldots, H - 1 \}$. Here, $\Delta < 0$ corresponds to the ``observed'' duration used for calculating the projection matrix $\bm{A}_t$, and $\Delta \geq 0$ corresponds to the ``unseen'' duration generated by extrapolation. From this, we observe the followings:
\begin{itemize}
    \item \textbf{Monotonicity.} Increasing the number $C$ of eliminated principal components can reduce the volatility of the spectral residuals even for the unseen duration.
    %\minami{I think monotonicity in in-sample duration is trivial, so reviewers may misunderstand this point.}
    \item \textbf{Local stability.} Regarding the volatility proportion, the difference between the observed duration and the unseen duration increases with increasing $C$. Remarkably, the first principal component (a.k.a. the market factor) is quite stable, and the volatility proportion for $C = 1$ is well extrapolated to the unseen duration.
    %\item Residual factors that eliminates more principal components have larger proportional difference between observed data and unseen data.
\end{itemize}
The above observations suggest that the projection matrix $\bm{A}_t$ can be locally stable if $C$ is not too large. Hence, we can extract meaningful residual information in the subsequent time window by extrapolating the projection matrix.

\subsection{Choosing the number of eliminated factors}\label{sec:ex-choosing-c}

\begin{table}[t]
    \centering
    \caption{Performance comparison of reversal returns on different numbers of principal components (PCs) to be eliminated (U.S.~market).}
    \scalebox{0.85}{
    \begin{tabular}{@{}lcccccc@{}}
        \toprule
        
         & ASR$\uparrow$ & AR$\uparrow$ & AVOL$\downarrow$ & DDR$\uparrow$ & CR$\uparrow$ & MDD$\downarrow$ \\ \midrule
         
        % US-Data
        \textbf{$C = 0$} & +0.759 & \textbf{+0.076} & 0.058 & +1.325 & +0.367 & 0.158 \\
        \textbf{$C = 1$} & +0.753 & +0.054 & 0.041 & +1.343 & +0.309 & 0.132 \\
        \textbf{$C = 10$} & \textbf{+1.426} & +0.035 & 0.049 & \textbf{+2.541} & \textbf{+1.206} & 0.041 \\
        \textbf{$C = 20$} & +1.317 & +0.028 & 0.037 & +2.288 & +1.000 & 0.037 \\
        \textbf{$C = 50$} & +1.264 & +0.019 & 0.024 & +2.275 & +0.990 & \textbf{0.025} \\
        \textbf{$C = 100$} & +1.089 & +0.013 & \textbf{0.014} & +1.877 & +0.391 & 0.035 \\
        
        \bottomrule
    \end{tabular}
    }
    \label{table-reversal-returns-with-pca}
\end{table}
\begin{table}[t]
    \centering
    \caption{Performance comparison of reversal returns on different numbers of principal components (PCs) to be eliminated (Japanese market).}
    \scalebox{0.85}{
    \begin{tabular}{@{}lcccccc@{}}
        \toprule
         & ASR$\uparrow$ & AR$\uparrow$ & AVOL$\downarrow$ & DDR$\uparrow$ & CR$\uparrow$ & MDD$\downarrow$ \\ \midrule
        
        $C = 0$ & +1.259 & +0.082 & 0.065 & +2.344 & +0.700 & 0.117 \\
        $C = 1$ & +1.715 & \textbf{+0.087} & 0.050 & +3.169 & +1.103 & 0.078 \\
        $C = 10$ & +2.960 & +0.074 & 0.025 & +5.555 & +2.037 & 0.036 \\
        $C = 20$ & \textbf{+3.291} & +0.071 & 0.022 & \textbf{+6.306} & \textbf{+4.215} & 0.017 \\
        $C = 50$ & +3.036 & +0.047 & 0.016 & +5.938 & +3.258 & \textbf{0.015} \\
        $C = 100$ & +2.175 & +0.023 & \textbf{0.011} & +4.114 & +1.152 & 0.020 \\
        
        \bottomrule
    \end{tabular}
    }
    \label{table-reversal-returns-with-pca-jp}
\end{table}

%Thirdly, 
Next, we consider the appropriate choice of the parameter $C$ to make trading strategies robustly profitable. Figure \ref{fig-reversal-returns} and Table \ref{table-reversal-returns-with-pca} show the performance of the reversal strategy over the spectral residuals on U.S.~market.
%Here, we also assumed one-day delay for executing the strategies.

%\minami{Replace the following part with the S\&P 500 experiment. I guess the number $C = 50$ and the cumulative return results will change. Then rewrite this paragraph accordingly.}
From the result, we observe the followings.
%\minami{TODO: Show actial values of AR and AVOL}
With increasing the number $C$ of eliminated principal components, both AR and AVOL tend to decrease. This is because eliminating more principal components reduces the volatility, while it becomes more difficult to earn large returns from the remaining information.
Especially, eliminating the first $10$ principal components ($C=10$) attained the best trade-off between the return and the risk, and thus had the best ASR value.

We also conducted a similar experiment on Japanese market data, in which $C = 20$ achieved the best ASR value. See Figure \ref{fig-reversal-returns-jp} and Table \ref{table-reversal-returns-with-pca-jp} for corresponding results.

\begin{figure*}
    \begin{minipage}[t]{.48\textwidth}
        \centering
        \includegraphics[width=0.9\textwidth]{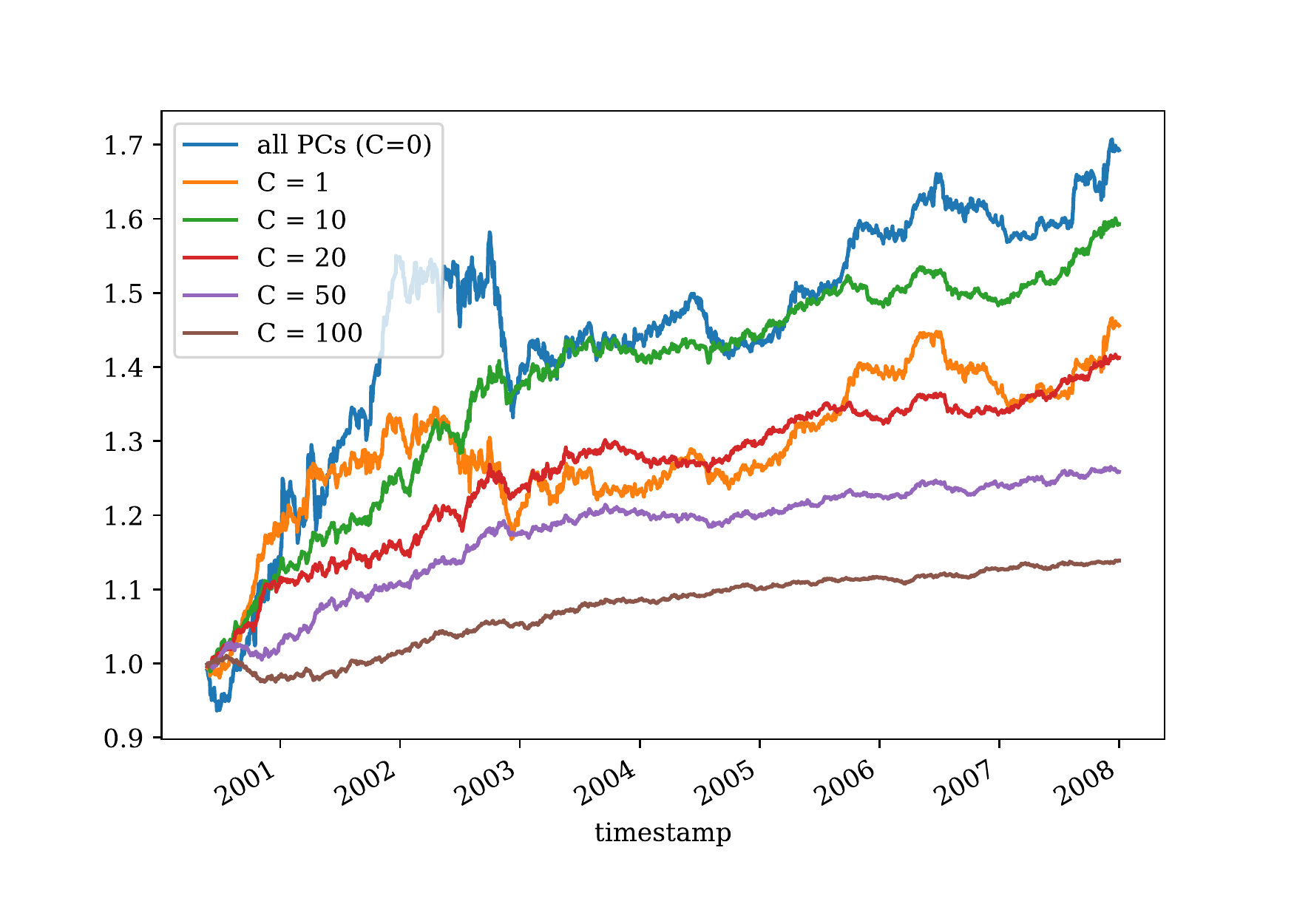}
        \caption{The Cumulative Wealth of the reversal strategy with different choices of the number $C$ of principal components to be eliminated (U.S.~market).}
        \label{fig-reversal-returns}
    \end{minipage}
    \hfill
    \begin{minipage}[t]{.48\textwidth}
        \centering
         \includegraphics[width=0.9\textwidth]{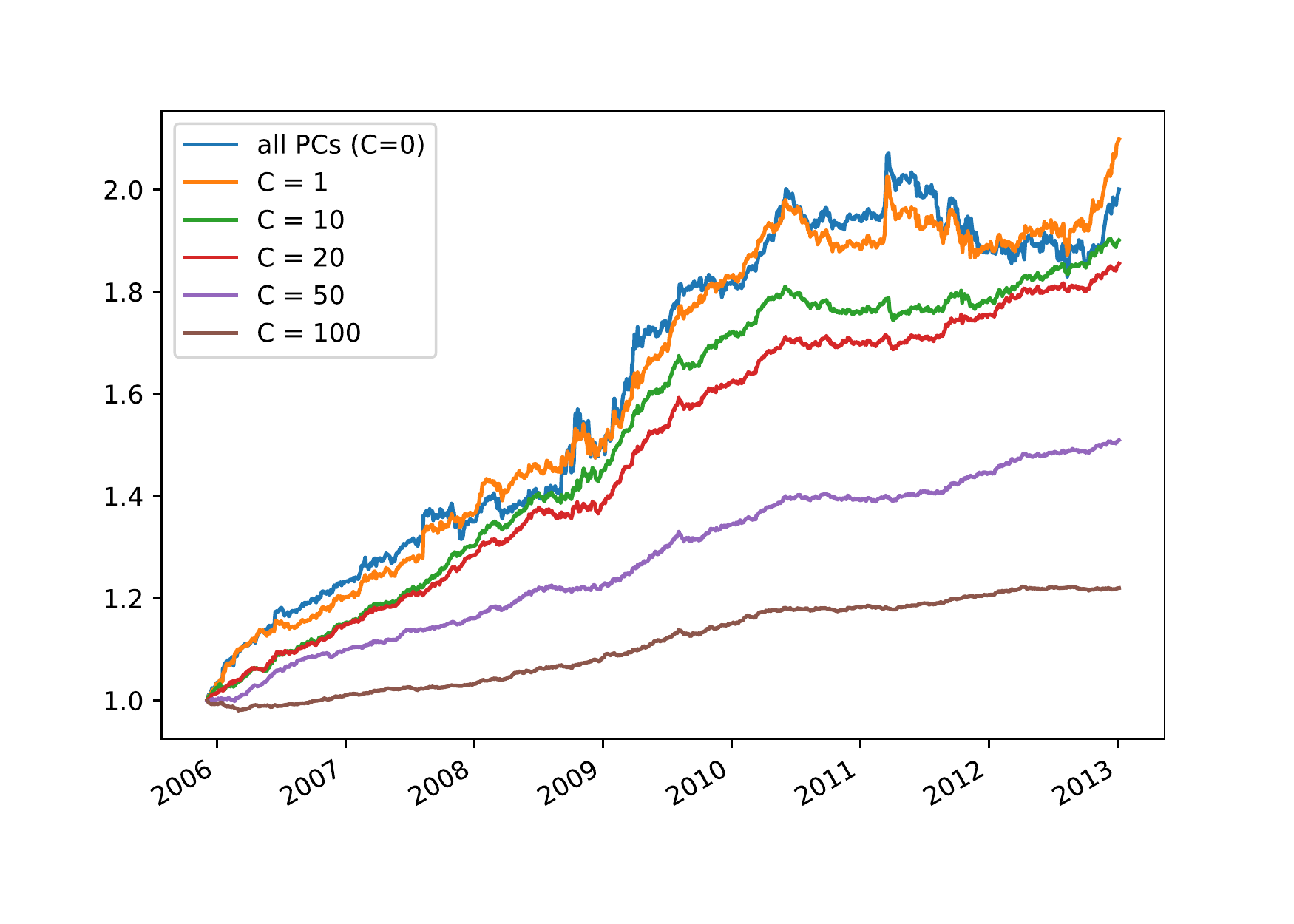}
        \caption{The Cumulative Wealth of the reversal strategy with different choices of the number $C$ of principal components to be eliminated (Japanese market).}
        \label{fig-reversal-returns-jp}
    \end{minipage}
\end{figure*}

\section{Performance evaluation on Japanese market data}\label{sec:japanese}

Here, we evaluated the performance of our proposed system (DPO) on Japanese market data. We used similar evaluation metrics and baseline methods presented in Section 4.3.

\subsection{Japanese market data}

For Japanese market data, we used daily prices of stocks listed in TOPIX 500 from January 2005 to December 2018. We used data before January 2012 for training and validation and the remainder for testing. We obtained the data from Japan Exchange Group (JPX)\footnote{\url{https://www.jpx.co.jp/english/markets/}}.
%For both of the above datasets,

\subsection{Results}

\begin{table}[t]
    \centering
    \caption{Performance comparison on the Japanese market.}
    \label{table-final-jp}
    {\scriptsize
    \begin{tabular}{@{}lccccccc@{}}
        \toprule
        & ASR$\uparrow$ & AR$\uparrow$ & AVOL$\downarrow$ & DDR$\uparrow$ & CR$\uparrow$ & MDD$\downarrow$ \\ \midrule
        
        \textbf{Market} & +0.819 & \textbf{+0.158} & 0.193 & +1.320 & +0.636 & 0.248 \\
        \textbf{AR(1)} & +1.511 & +0.058 & 0.038 & +2.684 & +1.094 & 0.053 \\ \midrule
        \textbf{AR(1) on SRes} & +1.835 & +0.034 & 0.019 & +3.237 & +1.112 & 0.031 \\
        \textbf{Linear on SRes} & +1.380 & +0.023 & \textbf{0.017} & +2.469 & +0.952 & \textbf{0.024} \\
        \textbf{MLP on SRes} & +1.802 & +0.030 & \textbf{0.017} & +3.280 & +0.963 & 0.032 \\
        \textbf{SFM on SRes} & +0.250 & +0.005 & 0.019 & +0.442 & +0.116 & 0.042 \\
        \midrule
        \textbf{DPO-NQ} & +1.369 & +0.024 & 0.018 & +2.379 & +0.758 & 0.032 \\
        \textbf{DPO-NF} & +1.770 & +0.030 & \textbf{0.017} & +3.159 & +1.165 & 0.025 \\
        \textbf{DPO-NV} & +1.979 & +0.039 & 0.020 & +3.516 & \textbf{+1.484} & 0.026 \\
        \textbf{DPO} & \textbf{+2.171} & +0.036 & \textbf{0.017} & \textbf{+3.878} & +1.460 & 0.025 \\
        
        \bottomrule
    \end{tabular}
    }
\end{table}

% \begin{table}[t]
%     \centering
%     \caption{Performance comparison on Japanese market without spectral residual extraction.
%     \minami{Add Japanese version!}
%     }
%     \label{table-final-jp-raw}
%     {\scriptsize
%     \input{04_tb03_final_us_raw}
%     }
% \end{table}

\begin{figure}[t]
    \begin{minipage}[t]{.5\textwidth}
        \centering
        \begin{tikzpicture}
            \node[inner sep=0] (I) at (0, 0)
            {\includegraphics[width=0.90\textwidth]{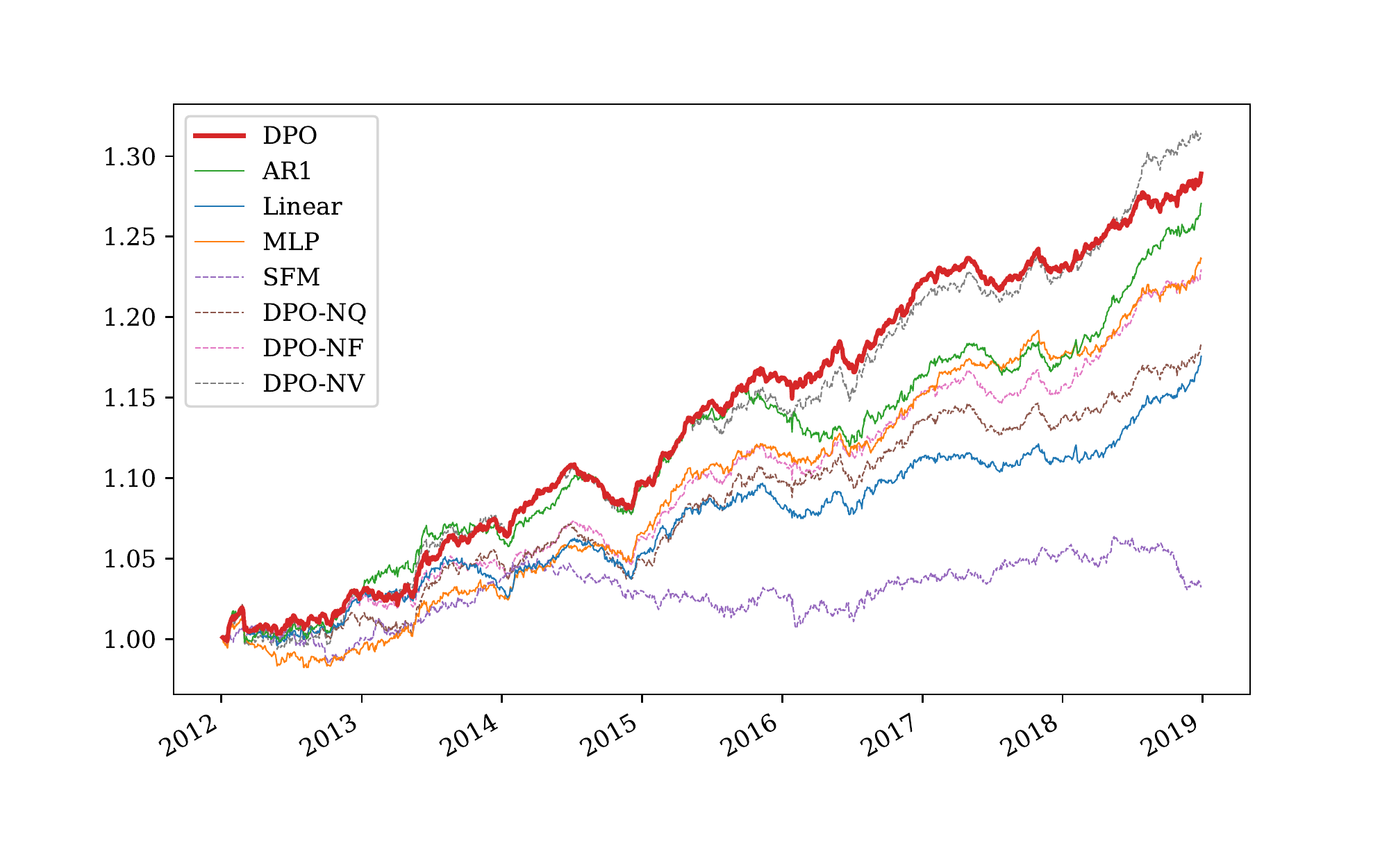}};
            \node[left=0cm of I] {\rotatebox{90}{Cumulative wealth}};
            \node[below=0cm of I] {Date};
        \end{tikzpicture}
        \caption{The Cumulative Wealth in Japanese market.}
        \label{fig-final-ablation-jp}
    \end{minipage}
\end{figure}

Table \ref{table-final-jp} and Figure \ref{fig-final-ablation-jp} show the results. Overall, we obtained consistent results with those on U.S.~market. Our proposed method (DPO) outperformed other baseline methods and ablation models in the ASR. For other evaluation metrics, DPO achieved comparable performance to the best performing baselines.

\end{document}